\documentclass[11pt]{article}
\usepackage{fullpage}
\usepackage{graphicx}
\usepackage{rotating}
\usepackage{amsthm}
\usepackage[T1]{fontenc}
\usepackage{enumerate}
\usepackage{amssymb}
\usepackage{amsmath}
\usepackage{fancyhdr}
\usepackage{algorithm, algorithmic}
\usepackage{url}
\usepackage{verbatim}
\usepackage{mathtools}
\usepackage{soul}

\usepackage{xcolor}

\usepackage{times}

\floatname{algorithm}{Algorithme}

\let\mylistof\listof
\renewcommand\listof[2]{\mylistof{algorithm}{Liste des algorithmes}}

\makeatletter
\providecommand*{\toclevel@algorithm}{0}
\makeatother
\newtheorem{theoreme}{Theorem}
\newtheorem{lemma}{Lemma}
\newtheorem{observation}{Observation}
\newtheorem{proposition}{Proposition}
\newtheorem{conjecture}{Conjecture}
\newtheorem{question}{Open question}

\theoremstyle{remark}
\newtheorem{remark}{Remark}

\newcommand{\indet}{\fbox{$\scriptstyle{?}$}}
\newcommand{\plus}{\fbox{$\scriptstyle{+}$}}
\newcommand{\moins}{\fbox{$\scriptstyle{-}$}}
\newcommand{\spm}{\fbox{$\scriptstyle{\pm}$}}

\def\P{\mathcal{P}}
\def\E{\mathcal{E}}
\def\B{\mathcal{B}}
\def\C{\mathcal{C}}
\def\det{det}

\def\R{\mathbb{R}}

\title{Orientations of Simplices\\Determined by Orderings 
on the Coordinates of their Vertices%
\thanks{A short preliminary conference version of this paper has been published \cite{cccg}.}
\thanks{This research is part of the OMSMO Project
 (Oriented Matroids for Shape Modeling), 
supported by the ``Chercheur d'avenir'' Languedoc-Roussillon  Grant  and the ``Fonds europ\'een de d\'eveloppement r\'egional'' FEDER. Supported formerly by  the TEOMATRO 
Grant ANR-10-BLAN 0207.}
}

\author{

Emeric Gioan\footnotemark[1] \and 

Kevin Sol\footnotemark[2]
\and 

G\'erard Subsol\footnotemark[1]
}

\font\eightrm=cmr10 scaled 800
\date{\eightrm \today}

\begin{document}

\maketitle

%\footnotetext[1]{CNRS}
%\footnotetext[2]{LIRMM, University of Montpellier 2, France. {\it Email:} {\tt \{lastname\}@lirmm.fr}}

\footnotetext[1]{CNRS, LIRMM, Universit\'e de Montpellier, France. {\it Email:} {\tt \{lastname\}@lirmm.fr}}
\footnotetext[2]{This research was done when Kevin Sol was Ph.D. at the LIRMM (research teams AlGCo and ICAR), Universit\'e de Montpellier, France}

%%%%%%%%%%%%%%%%%%%%%%%%%%%%%%%%%%%%%%%%%%%%%%
\vspace{-0.5cm}

\def\corr#1{ \textcolor{blue}{\tiny [#1]}} % CORRECTION GRAMMATIXCALE OU DE TOURNURE

\def\corp#1{ \textcolor{purple}{\tiny [#1]}} % CORRECTION PETIT MOT A LA CON

\def\corv#1{ \textcolor{brown}{\tiny [#1]}} % CORRECTION MOT MANQUANT

\def\cora#1{ \textcolor{red}{\tiny [#1]}} % VARIANTE OK

\def\AFINIR{\textcolor{red}{******}}

\begin{abstract}
Provided $n$ points in an $(n-1)$-dimensional affine space, and one ordering of the points for each coordinate,
we address the problem of testing whether these orderings determine if the points are the vertices of a simplex (i.e. are affinely independent),  regardless of the real values of the coordinates. 
We also attempt to determine the orientation of this simplex.
  In other words, given a matrix whose columns correspond to affine points, we want to know when the sign (or the non-nullity) of its determinant is implied by orderings given to each row for the values of the row.
We completely solve the problem in dimensions 2 and 3.
%(i.e. for $n=3$ and $n=4$), 
We provide a direct combinatorial characterization, along with 
%an algorithm 
a formal calculus method. It can also be viewed as a decision algorithm, and is based on testing the existence of a suitable inductive cofactor expansion  of the determinant. We conjecture that our
%formal calculus 
method generalizes in higher dimensions.
This work aims
to be part of a study on how oriented matroids encode shapes of 3-dimensional landmark-based objects. Specifically, applications include the analysis of anatomical data for physical anthropology and clinical research. 
\smallskip


\smallskip
%\bigskip
\end{abstract}

%%%%%%%%%%%%%%%%%%%%%%%%%%%%%%%%%%%%%%%%%%%%%%%

\section{Introduction}
\label{sec:intro}

We consider $n$ points in an $(n-1)$-dimensional real affine space. 
For each of the $n-1$ coordinates, an ordering is given and applied to the $n$ values of the points with respect to this coordinate.
We address the problem of testing if these points are the vertices of a simplex (i.e. are affinely independent, i.e. do not belong to a same hyperplane), and of determining the orientation of this simplex, 
assuming that the coordinates of the points satisfy the given orderings,
 independently of their real values.

More formally, we consider the following generic matrix
(where each $e_i$ is the label of a point, forming the set $\E$, and each $b_i$ is the index of a coordinate, forming the set $\B$)
\begin{displaymath}
M_{\E,\B}=\left(\begin{array}{c c c c}
1&1&\ldots&1\\
x_{e_1,b_1}&x_{e_2,b_1}&\ldots&x_{e_n,b_1}\\
x_{e_1,b_2}&x_{e_2,b_2}&\ldots&x_{e_n,b_2}\\
%\vdots&\vdots&\ddots&\vdots\\
\vdots&\vdots& &\vdots\\
x_{e_1,b_{n-1}}&x_{e_2,b_{n-1}}&\ldots&x_{e_n,b_{n-1}}
\end{array}\right)
\end{displaymath}
together with orderings given to the values for each row; we want to know when the sign (or the non-nullity) of the matrix determinant is determined by these orderings only.

Equivalently, we consider the above formal matrix and the affine algebraic variety of $\R^{n\times (n-1)}$ whose equation is $det(M_{\E,\B})=0$.
Then,
among (open) regions of $\R^{n\times (n-1)}$ delimited by the hyperplanes 
$x_{e_i,b_k}=x_{e_j,b_k}$ 
%$x_{i,k}=x_{j,k}$ 
for all $1\leq i,j\leq n$ and all $1\leq k\leq n-1$, 
we attempt to identify those having a non-empty intersection with this variety (obviously, regions delimited by these hyperplanes are in canonical bijection with coordinate linear orderings).

 We completely solve the problem in dimensions 2 (Section \ref{sec:dim-2}) and 3 (Section \ref{sec:dim-3}).
We provide a direct combinatorial characterization to test if the orientation is determined or not, along with a combinatorial formal calculus method which can also be viewed as a decision algorithm.
More precisely, our method is based on testing the existence of a suitable inductive cofactor expansion of the determinant, 
which allows the computation of the determinant sign using a combinatorial formal calculus. 
We conjecture that this formal characterization
%purely combinatorial inductive characterization
%using only some symbolic computation, 
generalizes in higher dimensions (Section \ref{sec:tools}). 
%Still, for the sake of our applications, we just needed to solve the problem in dimensions 2 and 3.

%V23 ajout paragraphe ci-dessous sur les SNS
Interestingly, the problem addressed here is formally close to the classical problem of 
sign nonsingular matrices (SNS), see \cite{SNS}. However, the two problems are rather separate. Let us explain this briefly.
The common feature of the two problems relies in the following situation.
Consider a square $n\times n$ matrix $N$ whose entries are signs in $\{+, -\}$.
 From our setting, such an $n\times n$ matrix $N$ can be obtained naturally from $M_{\E,\B}$ and a linear ordering for each row by subtracting a column - say $e_i$ - from every other column, and replacing entries with their signs w.r.t. the linear orderings. Thus, reciprocally,  the sign data in such a matrix $N$ is  interpreted in our setting as an ordering relation for each row
of type: ( the set of $-$ ) $<$ ( the set of $+$ ), corresponding to:
 ( a set $A\subset \E$ ) $<$ $e_i$ $<$ ( a set $B\subset \E$ ).
%
%The question, in both settings, is:
The question is:
assuming real values are assigned to entries of the matrix $N$, such that these values have the same signs as the signs in the matrix,
is this matrix always invertible?
To this particular question, the answer is always NO, whatever the signs, unless $n\leq 2$ (see \cite[page 108]{SNS}: an SNS-matrix of order $n\geq 3$ has at least one zero entry; see also Remark \ref{rk:SNS}). However, in more general settings, the answer can be YES.
The SNS setting and ours consist in two different variants of the above question.
They both yield a non-trivial question and provide interesting classes of sign patterns.
In the SNS setting, the variant is to consider the same question with signs in $\{+, - , 0\}$ instead of $\{+, -\}$.
%There are variants of this question (Strong SNS involving the inverse matrix, non-square L-matrices...) but they rely fundamentally on the positions of zeros in the sign matrix.
%
In our setting,
we specify the question keeping signs  in $\{+,-\}$ while restricting
the available real values to values satisfying more involved ordering relations for each row
(between all elements and not only between two subsets $A$ and $B$).
There seems to be no obvious connection between the two problems.
Indeed, the zeros in the SNS setting
and the linear orderings in ours place significantly different constraints upon the sets of real values to be tested.
%\AFINIR EME-***************REF A BRUALDI

Finally, we point out that this work aims  to be part of a general study on how oriented matroids \cite{OM99} encode shapes of 3-dimensional landmark-based objects. 
Specifically, applications include the analysis of anatomical data for physical anthropology and clinical research \cite{miccai}\cite{poster}. 
%V23 ajout \cite{miccai} ci-dessus
%
%Motivations for this work are to study how oriented matroids \cite{OM99} encode shapes of 3-dimensional objects, with applications in particular to the analysis of anatomical data for physical anthropology and clinical research \cite{poster}.
%  
In these applications, we usually study  a set of models belonging to  a  given group (e.g. sets  of 3D landmark points located on human or primate skulls) and we search for the significant properties encoded by the combinatorial structure.
Our proposed solution allows us to distinguish chirotopes (i.e. simplex orientations) which are determined by the model's ``generic'' form  (e.g. in any skull, the mouth is below the eyes) from those which are specific to anatomical variations.
Examples of 3D anatomical data results are presented in Sections~\ref{sec:example} and \ref{sec:ex-suite}.

%%%%%%%%%%%%%%%%%%%%%%%%%%%%%%%%%%%%%%%%%%%%%
%%%%%%%%%%%%%%%%%%%%%%%%%%%%%%%%%%%%%%%%%%%%%%%%%%%%

\section{Preliminaries}

We warn the reader that we purposely use rather abstract formalism throughout the paper (formal variables instead of real values, indices within arbitrary ordered sets instead of integers). This will allow us to get simpler and non-ambiguous constructions and definitions.

%\section{Formalism and terminology of the problem}
\subsection{Formalism and terminology of the problem}

Let us fix an (ordered) set $\E=\{e_1,\dots,e_n\}$, with size $n$, of \emph{labels}, and an (ordered) canonical basis $\B=\{b_1,\dots,b_{n-1}\}$, with size $n-1$, of the $(n-1)$-dimensional real space $\R^{n-1}$.
We denote $M_{\E,\B}$ - or $M$ for short when the context is clear - the formal matrix 
\begin{displaymath}
M_{\E,\B}=\left(\begin{array}{c c c c}
1&1&\ldots&1\\
x_{e_1,b_1}&x_{e_2,b_1}&\ldots&x_{e_n,b_1}\\
x_{e_1,b_2}&x_{e_2,b_2}&\ldots&x_{e_n,b_2}\\
%\vdots&\vdots&\ddots&\vdots\\
\vdots&\vdots& &\vdots\\
x_{e_1,b_{n-1}}&x_{e_2,b_{n-1}}&\ldots&x_{e_n,b_{n-1}}
\end{array}\right)
\end{displaymath}
whose entry in column $i$ and row $j+1$, for $1\leq i\leq n$ and $1\leq j\leq n-1$, is the formal variable $x_{e_i,b_j}$.
%, standing as a set of labels for $n$ real points.
%
The determinant $det(M_{\E,\B})$ of this formal matrix is a multivariate polynomial in these formal variables, and the main object studied in this paper.

Let $\mathcal{P}$ be a set of $n$ points, labeled by $\E$,  in $\R^{n-1}$ considered as an affine space. 
We denote $M_{\E,\B}(\P)$ - or $M(\P)$ for short - the matrix whose columns give the coordinates of points in $\P$ w.r.t. the basis $\B$. This comes down to specifying real values for the formal variables $x_{e_i,b_j}$ in the matrix $M_{\E,\B}$ above.
For $e\in\E$ and $b\in\B$, we denote $x_{e,b}(\P)$ the real value given to the formal variable  $x_{e,b}$ in $\P$.
We may sometimes denote $x_{e,b}$ for short instead of $x_{e,b}(\P)$ when the context is clear.
%*** PHRASE PRECEDENTE PEUT-ETRE A ENLEVER ET A REMETTRE JUSTE QUAND ON LE FAIT DANS LES PREUVES ***
%We may sometimes make the classical notation abuse to denote with the same symbol $x_{e_i,b_j}$ either the formal variable, or its real value when $\P$ is specified, but the distinction will be clear form the context.
%%With the notation of $M_\P$ given above: $x_{i,j}$ is the $j$-th coordinate of the  point with label $i$ in $\mathcal{P}$.
%
We call \emph{orientation} of $\mathcal{P}$, or \emph{chirotope} of $\mathcal{P}$ in the oriented matroid terminology,  the sign of $\det(M({\mathcal{P}}))$, belonging to the set $\{+,-,0\}$. It is the sign of the real evaluation of the polynomial $det(M)$ at the real values given by $\P$. This sign is not equal to zero if and only if $\mathcal{P}$ forms a \emph{simplex} (basis of the affine space).
%
%Si $\det(M_{\mathcal{P}})=0$ alors l'orientation de la base formé par $\mathcal{P}$ correspond au fait que les points de $\mathcal{P}$ appartiennent à un même hyperplan. Sinon, l'orientation de la base formée par $\mathcal{P}$ correspond à l'orientation du simplexe formé par les points de $\mathcal{P}$.

%To compute $\det(M_{\mathcal{P}})$ we will need to transform the matrix $M_{\mathcal{P}}$. On notera par exemple $M_{\mathcal{P}}\xRightarrow{C3\leftarrow C3-C1}M$ pour indiquer qu'on obtient $M$ en soustrayant la première colonne à la troisième colonne dans $M_{\mathcal{P}}$. On a alors $\det(M_{\mathcal{P}})=\det(M)$. Pour pouvoir développer simplement le déterminant selon la première ligne, nous nous intéresserons particulièrement aux matrices pour lesquelles on a fait des soustractions entre les colonnes de sorte que toutes les colonnes sauf une (qui elle ne change pas) deviennent alors de la forme suivante :
%\begin{displaymath}
%\left(\begin{array}{c}
%0\\
%x_{b,1}-x_{a,1}\\
%x_{b,2}-x_{a,2}\\
%\vdots\\
%x_{b,n-1}-x_{a,n-1}
%\end{array}\right)
%\end{displaymath}
%où $a$ et $b$ sont 2 points de $\mathcal{P}$.
%%%%% On dira qu'une telle matrice est \emph{équivalente à $M$}.
%
%Quand il n'y aura pas d'ambiguïté, on utilisera la notation $x_{b-a,i}$ pour représenter $x_{b,i}-x_{a,i}$.
%
%As we will transform the matrix by making operations on its columns, we will denote $x_{b-a,i}$ instead of $x_{b,i}-x_{a,i}$, where $a$ and $b$ are two elements of $\P$ (not to be considered as integers).

%\begin{defini}
%V23 'set' -> 'list' ci-dessous
We call  \emph{ordering configuration on $(\E,\B)$} - or \emph{configuration} for short - a list $\C$ of $n-1$ orderings $<_{b_1}, \ldots, <_{b_{n-1}}$ on $\E$, with one ordering for each element of $\B$.
%We call  \emph{configuration of $n-1$ orderings on }$\E$ a sequence of $n-1$ orderings $<_1 \ldots <_{n-1}$ on $\E$.
%\end{defini}
%
%We will usually denote $\mathcal{C}$ a configuration on $\mathcal{E}$, and $\mathcal{C}_T$ a configuration where each ordering is linear, which we call  a \emph{linear ordering configuration}.
In general, such an ordering can be any partial ordering.
If every ordering $<_b$, $b\in \B$, is linear, then $\C$ is called 
 a \emph{linear ordering configuration}.
An element of $\E$ which is the smallest or the greatest in a linear ordering on $\E$ is called \emph{extreme} in this ordering.
We call \emph{reversion} of an ordering the  ordering obtained by reversing every inequality in this ordering.

%\begin{defini}
Given a configuration $\C$ on $(\E,\B)$ and a set of $n$ points $\P$ labeled by $\E$, we say that $\P$ \emph{satisfies} ${\C}$ if, for all $b \in \B$, the natural order (in the set of real numbers $\mathbb{R}$) of the coordinates $b$ of the points in $\mathcal{P}$ is compatible with the ordering $<_b$ of ${\C}$, that is 
precisely : $$\forall b \in \B, \ \ \forall e, f \in \E,\ \ e <_b f \ \Rightarrow \ x_{e,b}(\P) < x_{f,b}(\P).$$
%\end{defini}
%
%\begin{defini}
%Let ${\C}$ be a configuration on $(\E,\B)$. 
%
One may observe that the set of all $\P$ satisfying $\C$ forms a %convex
convex polyhedron, or more precisely: a (full dimensional) region of the space $\mathbb{R}^{n\times (n-1)}$, delimited by hyperplanes of equations of type $x_{e,b} = x_{f,b}$ for $b \in \B$ and $e, f \in \E$.

We say that a configuration ${\C}$ is \emph{fixed} if all the sets of points $\mathcal{P}$ satisfying ${\C}$ form a simplex and have the same orientation. 
%In this case, the sign of $det(M(\P))$ does not depend on $\P$ satisfying $\C$, 
In this case, the sign of $det(M(\P))$ is the same for all $\P$ satisfying $\C$.
Then we call \emph{sign of $det(M)$} 
%$det_\C(M)$ 
this sign, belonging to $\{\plus,\moins\}$ accordingly,
and we denote it $\sigma_\C(det(M))$.
If ${\C}$ is \emph{non-fixed}, then its \emph{sign} is
%det_\C(M)=\spm$.
$\sigma_\C(det(M))=\spm$.
%Conversely, ${\C}$ is \emph{non-fixed} if there exist two sets of points $\mathcal{P}_1$ and $\mathcal{P}_2$ satisfying ${\C}$ that do not have the same orientation. In this case, we denote $det_\C=\pm$.
%\end{defini}

%Observe 
%%also 
%that a symmetry or a rotation (a permutation of the set of coordinates $\B$), as well as a relabelling (a permutation of the set of labels $\E$), does not change the property of being fixed or not.
%Also, 
%The following observation comes directly from the convexity of the sets of $\P$ satisfying $\C$ and the continuity of the determinant.

%\begin{observation}
\begin{lemma}
\label{lem:non-fixed}
The following propositions are equivalent:

(a) The configuration $\C$ is non-fixed, that is $\sigma_\C(det(M))=\spm$. %$det_\C(M)=\spm$;

(b) There exist two sets of points $\mathcal{P}_1$ and $\mathcal{P}_2$ satisfying ${\C}$ and forming simplices that do not have the same orientation, that is $det(M(\P_1))>0$ and $det(M(\P_2))<0$;

(c) There exists a set of points $\mathcal{P}$  satisfying ${\C}$ and such that  the points of $\mathcal{P}$ belong to one hyperplane, that is $det(M(\P))=0$.
%\endofproof
%\end{observation}
\end{lemma}

\begin{proof}
By definition we have a) if and only if b) or c) is true.
The region of the space $\mathbb{R}^{n\times (n-1)}$ whose elements $\P$ satisfy $\C$ is a convex and, topologically, an open set of points in $\mathbb{R}^{n\times (n-1)}$.
So b) implies c) by  convexity  and  continuity of the determinant.
Moreover, c) implies b) %by this topological property, 
since, given $\P$ in this region such that $det(M(\P))=0$, one can add a matrix small enough to $M(\P)$ to get $\P'$ in the same region and such that $det(M(\P'))>0$, or also such that $det(M(\P'))<0$.
\end{proof}

%\begin{lemma}
%\conf
%One may easily notice that ${\C}$ is a non-fixed configuration if and only if there exist two sets of points $\mathcal{P}_1$ and $\mathcal{P}_2$ satisfying ${\C}$ and forming simplices that do not have the same orientation (that is $det(M(\P_1))=+$ and $det(M(\P_2))=-$), if and only if there exists a set of points $\mathcal{P}$ satisfying ${\C}$ and such that the points of $\mathcal{P}$ belong to one hyperplane (that is $det(M(\P))=0$).
%
%Observe 
%%also 
%that a symmetry or a rotation (a permutation of the coordinates), as well as a relabeling (a permutation of the labels), does not change the property of being fixed or not.
%\end{lemma}

%Two configurations on $(\E,\B)$ are called \emph{equivalent} if they are equal up to a permutation of $\B$ (symmetries and rotations in particular) and a permutation of $\E$ (relabelling). 
%
Two configurations on $(\E,\B)$ are called \emph{equivalent} if they are equal up to a permutation of $\B$,  a permutation of $\E$ (relabelling), and some reversions of orderings (geometrical symmetries). Note that, in a matricial setting, changing a configuration into an equivalent one comes down to changing the orderings of rows and columns, and to multiplying some rows by $-1$.
 Obviously, those operations do not change the non-nullity of the determinant. Hence, two equivalent configurations are fixed or non-fixed simultaneously.
\bigskip

Now, given an ordering configuration ${\C}$, the aim of the paper is to determine if ${\C}$ is fixed or non-fixed.
%\bigskip

%%%%%%%%%%%%%%%%%%%%%%%%%%%%%%%%%%%%%%%%%%%%%%%%%%%

%\section{An example from applications}
\subsection{An application example}
\label{sec:example}
%\noconf
%{\it Example.} ***verifier coherence avec notations precedentes***

Let us consider ten anatomical landmark points  in $\mathbb{R}^3$ chosen by experts on the 3D model of a skull 
from \cite{Crane}, as shown in Figure~\ref{10ptsCrane}.
We choose a canonical basis $(O,\vec{x},\vec{y},\vec{z})$ such that the axis 
$\vec{x}$ goes from the right of the skull to its left, the axis $\vec{y}$ goes from the bottom of the skull to its top, and the axis $\vec{z}$ goes from the front of the skull to its back.
%The specificity of this 3D model\corr{ as being a skull} of a skull \AFINIR 
This 3D model has the specificity of being a skull, which
implies that some coordinate ordering relations are satisfied by those points: for instance the point 9 (right internal ear) will always be on the right, above and behind with respect to point 5 (right part of the chin).
Figures \ref{10ptsFace} and \ref{10ptsProfil} show those points respectively from the front and from the right of the model, with a grid representing those coordinate ordering relations.
\medskip

\begin{figure}[h]
	\centering
		\includegraphics[height=5cm]{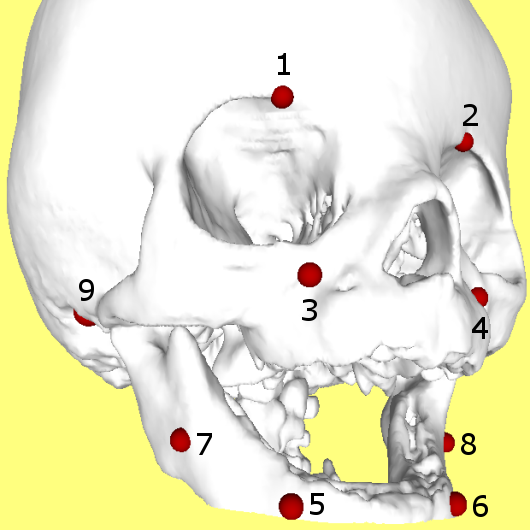}
	\caption{Ten anatomic points on a skull model \cite{Crane}}
	\label{10ptsCrane}
\end{figure}

\begin{figure}[h]

\begin{minipage}[c]{0.50\linewidth}
	\centering
		\includegraphics[height=4cm]{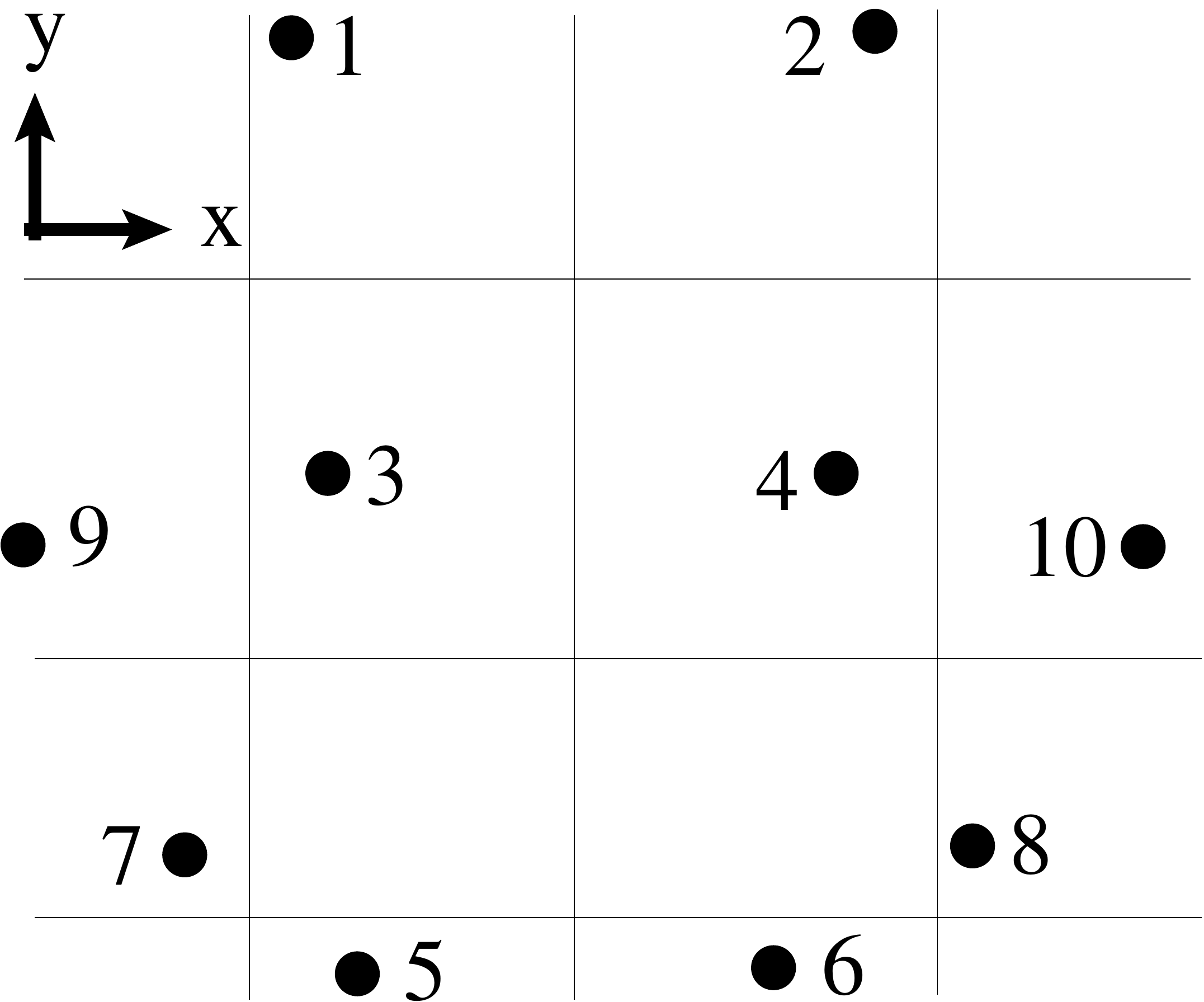}
	\caption{View from the front}
	\label{10ptsFace}
\end{minipage}
\begin{minipage}[c]{0.50\linewidth}
	\centering
		\includegraphics[height=4cm]{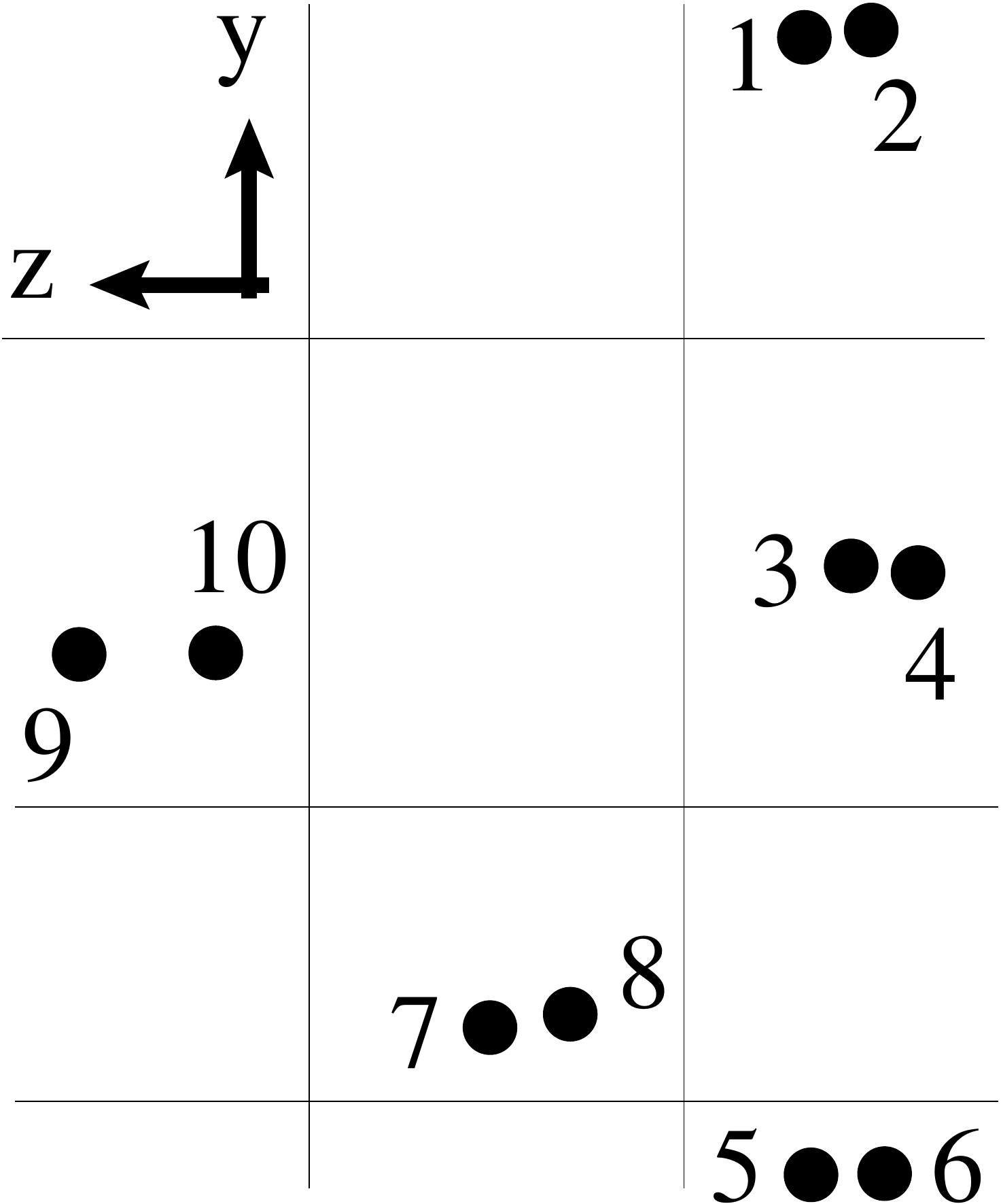}
	\caption{View from the right}
	\label{10ptsProfil}
\end{minipage}

\end{figure}

\begin{figure}[]
\begin{minipage}[b]{0.32\linewidth}
    \centering
        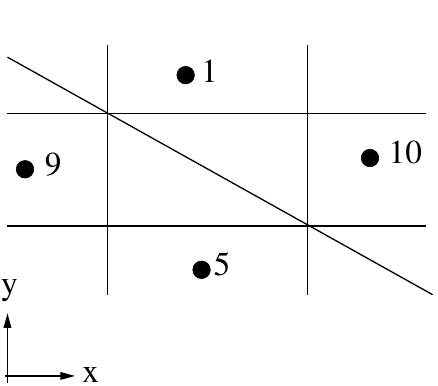
    \caption{View from the front}
    \label{casEmericXY}
\end{minipage}
\hfill
\begin{minipage}[b]{0.32\linewidth}
    \centering
        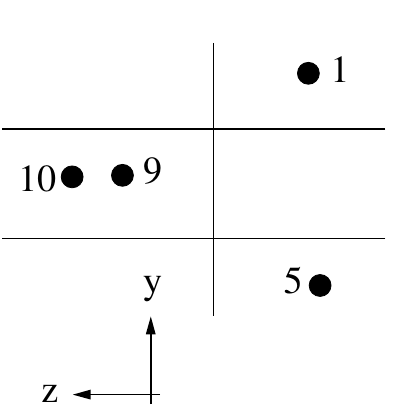
    \caption{View from the left}
    \label{casEmericYZ}
\end{minipage}
\hfill
\begin{minipage}[b]{0.32\linewidth}
    \centering
        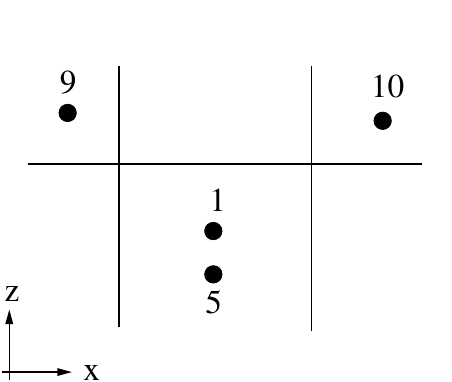
    \caption{View from below}
    \label{casEmericXZ}
\end{minipage}
\end{figure} 

%\begin{figure}[]
%\begin{minipage}[c]{0.45\linewidth}
%	\centering
%		\input{./figures/casEmericXY.pdf_t}
%	\caption{View from the front}
%	\label{casEmericXY}
%\end{minipage}
%\begin{minipage}[c]{0.45\linewidth}
%	\centering
%		\input{./figures/casEmericYZ.pdf_t}
%	\caption{View from the left}
%	\label{casEmericYZ}
%\end{minipage}
%\end{figure}

%\begin{figure}[!ht]
%\begin{minipage}[c]{0.50\linewidth}
%	\centering
%		\includegraphics[height=5cm]{./figures/10ptsFace.pdf}
%	\caption{Points vus de face}
%	\label{10ptsFace}
%\end{minipage}
%\begin{minipage}[c]{0.40\linewidth}
%	\centering
%		\includegraphics[height=5cm]{./figures/10ptsProfil.pdf}
%	\caption{Points vus de profil}
%	\label{10ptsProfil}
%\end{minipage}
%\end{figure}

For application purpose  (e.g. in \cite{miccai}\cite{poster}), 
%we are usually (e.g. in \cite{miccai}\cite{poster}) given\corr{ a set of} such models, 
we are given such models, 
%V23 ajout \cite{miccai} ci-dessus
coming from various individuals
(with possible pathologies) 
and species
(e.g. primates and humans), 
by physical anthropology and clinical research experts% of those fields 
who are interested
 in  mathematically characterizing and classifying them.
In this paper, our aim is to detect which configurations are fixed independently of the real values of the landmarks. 
%The meaning of these configurations is that 
These particular configurations are interesting to detect: they mean that the 
corresponding
relative positions of points 
%satisfying them 
%these configurations 
do not depend on some anatomical variabilities (e.g. on being a primate or a human skull), but only on the generic shape of the model (i.e. on being a skull).
%\AFINIR
%independently of the landmark coordinate real values.

\bigskip
The ordering configurations are represented in Figures \ref{10ptsFace} and \ref{10ptsProfil}, with $\E$ being any set of four points, and $\B$ corresponding to the three axis $\{x,y,z\}$.
As a preliminary example, let us consider the relations between  points labelled by $\E=\{1, 5, 9, 10\}$; we get the following configuration ${\C}$, as illustrated in Figures \ref{casEmericXY}, \ref{casEmericYZ} and \ref{casEmericXZ}:
\begin{displaymath}
\begin{array}{c}
9 <_x 1 <_x 10 \quad \text{ and } \quad 9 <_x 5 <_x 10\\
5 <_y 9 <_y 1 \quad \text{ and } \quad 5 <_y 10 <_y 1\\
1 <_z 9 \quad \text{ and } \quad 1 <_z 10 \quad \text{ and } \quad 5 <_z 9 \quad \text{ and } \quad 5 <_z 10
\end{array}
\end{displaymath}
%
%En ne conservant par exemple que les relations entre les points $\{1, 5, 9, 10\}$, on obtient la configuration ${\C}$ suivante :
%\begin{displaymath}
%\begin{array}{c}
%9 <_x 1 <_x 10 \quad \text{ et } \quad 9 <_x 5 <_x 10\\
%5 <_y 9 <_y 1 \quad \text{ et } \quad 5 <_y 10 <_y 1\\
%1 <_z 9 \quad \text{ et } \quad 1 <_z 10 \quad \text{ et } \quad 5 <_z 9 \quad \text{ et } \quad 5 <_z 10
%\end{array}
%\end{displaymath}
%
%Notre but serait alors de savoir si pour tout ensemble de points respectant cette configuration, l'orientation du simplexe formé par ces points est la même. Autrement dit, la configuration ${\C}$ est elle fixe ou non-fixe. Nous allons montrer dans le reste de cette partie que cette configuration est fixe. Pour ce faire, nous utiliserons un raisonnement géométrique. 
%
%
%We leave as an easy geometrical exercise to prove that this configuration is fixed. 
%\noconf
Proving that this configuration is fixed 
can be seen as a
%is an easy 
geometry exercise. 
%The idea is the following.
The sketch is the following.
Let us prove that line $(1,5)$ and line $(9,10)$ cannot intersect: this implies that the points cannot belong to a same hyperplane, and hence form a simplex with fixed orientation.
Consider two planes $\alpha_1$, $\alpha_2$ parallel to the directions $y,z$, two planes $\beta_1$, $\beta_2$ parallel to the directions $x,z$, and one plane $\gamma$ parallel to the directions $x,y$, consistent with the coordinate orderings, as shown in Figures \ref{casEmericXY}, \ref{casEmericYZ} and \ref{casEmericXZ}.
These 5 planes separate $\mathbb{R}^3$ in 18 regions. 
Consider the plane $\delta$ 
%containing the intersections of planes $\alpha_1,\beta_2,\gamma$ and $\alpha_2,\beta_1,\gamma$, 
containing the intersection of planes $\alpha_1,\beta_2,\gamma$, containing the intersections of planes $\alpha_2,\beta_1,\gamma$, 
and parallel to the direction $z$, as shown in Figure \ref{casEmericXY}.
Consider a region (among the 18 regions) intersecting $\delta$. 
Prove that, if line $(1,5)$ and line $(9,10)$ both intersect this region, then the two intersections are contained in two distinct parts of this region separated by $\delta$, meaning that the two lines do not intersect. 
The other cases (other regions) are either symmetric to this one or trivial.
%
%Consider the line $D$ contained in the plane $\gamma$ and spanned by the intersections of planes $\alpha_1,\beta_2,\gamma$ and $\alpha_2,\beta_1,\gamma$, then consider the plane $\delta$ orthogonal to $\gamma$ containing $D$, which cuts some regions in two. 
%
%Prove that, for every region of the space intersecting $\delta$, if both the line $(1,5)$ and the line $(9,10)$
%intersect this region, then each intersection is contained in two distinct parts of this region separated by $\delta$. A similar property holds for symmetric regions which may intersect lines $(1,5)$ and $(9,10)$. The conclusion is that $(1,5)$ cannot intersect $(9,10)$, which proves that the points cannot belong to a same hyperplane, and then form a simplex with fixed orientation.

%*** A VERIFIER, peut-etre a detailler davantage***
%***EVENTUELLEMENT A REPRENDRE MIEUX: preuve Kevin a alleger ou preuve Emeric a preciser (elles sont dans la source TEX sous ce paragraphe)***
\bigskip

In the rest of the paper, we develop tools to automatically detect fixed configurations,
 without having to use specific geometric constructions for each configuration as done above. Instead, our approach consists in unifying all configurations under a common combinatorial criterion. We will continue to study this example using  this approach in Section \ref{sec:ex-suite}.
%%%%%%%%%%%%%%%%%%%%%%%%%%%%%%%%%%%%%%%%%%%%%%%%%%%%%%%%%%%%%%%%%%%%%%%%%%%%%%%%%%%%

%%%%%%%%%%%%%%%%%%%%%%%%%%%%%%%%%%%%%%%%%%%%%%%%%%
%%%%%%%%%%%%%%%%%%%%%%%%%%%%%%%%%%%%%%%%%%%%%%%%%%
%%%%%%%%%%%%%%%%%%%%%%%%%%%%%%%%%%%%%%%%%%%%%%%%%%

%%%%%%%%%%%%%%%%%%%%%%%%%%%%%%%%%%%%%%%%%%%%
%%%%%%%%%%%%%%%%%%%%%%%%%%%%%%%%%%%%%%%%%%%%
%%%%%%%%%%%%%%%%%%%%%%%%%%%%%%%%%%%%%%%%%%%%
%\section{Combinatorial fixity criteria and conjectures}
%\section{Symbolic tools and fixity criteria}
\section{Computable fixity criteria and conjectures}
\label{sec:tools}

%%%%%%%%%%%%%%%%%%%%%%%%%%%%%%%%%%%%%%%%%%%%
%%%%%%%%%%%%%%%%%%%%%%%%%%%%%%%%%%%%%%%%%%%%
%%%%%%%%%%%%%%%%%%%%%%%%%%%%%%%%%%%%%%%%%%%%

%\section{From partial orderings to linear orderings}
\subsection{From partial orderings to linear orderings}

%\begin{defini}
%\label{extentsion_lineaire}
%We call a \emph{linear extension} of a set $E$ of $k$ orderings, a set of size $k$ where we replace all order of $E$ by one of its linear extensions. So a linear extension of a configuration $\mathcal{C}$ is a linear ordering configuration.
%\end{defini}
%
%***ATTENTION: je ne vois pas a quoi servait l'hypothese que deux elements doivent etre comparables, ni pourquoi on devrait restreindre les ordres a ces hypotheses... aussi j'ai raccourci les preuves, a voir... du coup cette section semble trop courte pour etre une section, pourrait-etre integree dans une autre section, peut-etre comme sous-section***

We recall that a linear extension of an ordering on a set $\E$ is a linear ordering on $\E$ compatible with this ordering.
A \emph{linear extension} of an ordering configuration $\C$ on $(\E,\B)$ is a linear ordering configuration on $(\E,\B)$ obtained by replacing each ordering on $\E$ in $C$ by one of its linear extensions.

%
%\begin{lemma}
%Let $\C$ be a configuration on $(\E,\B)$. If there exist 2 elements of $\E$ such that for all orderings of $\mathcal{C}$ these elements are not comparable, then $\C$ is non-fixed.
%\end{lemma}
%
%\begin{proof}
%Let $e$ and $f$ be the two elements of $\E$ that are not comparable in any ordering of $\C$. Let $\P$ be a set of points satisfying $\C$. We can choose 
%the points of $\P$ labeled by $e$ and $f$ to be equal. Then $det(M(\P))=0$ and $\C$ is non-fixed.
%%Notons $i$ et $j$ 2 éléments de $\mathcal{E}$ qui ne sont comparables dans aucun ordre de ${\C}$. Il existe alors un ensemble de points $\mathcal{P}$ respectant ${\C}$ pour lesquelles le point correspondant à $i$ et le point correspondant à $j$ ont toutes leurs coordonnées égales. Ainsi les points de $\mathcal{P}$ appartiennent à un même hyperplan ce qui montre que ${\C}$ est non-fixe.
%\end{proof}
%
%***ATTENTION je ne comprend pas pourquoi le lemme ci-dessus est utile, et je ne comprend pas pourquoi la preuve du lemme ci-dessous reposait sur l'hypothese que deux elements doivent etre comparables,
%je trouvais aussi cette preve inutilement longue, je l'ai changée, a verifier***

\begin{lemma}
\label{extLin}
Let $\C$ be a configuration on $(\E,\B)$.
%Let $\C$ be a configuration on $(\E,\B)$ such that two arbitrary elements of $\E$ are comparable in at least one ordering in $\C$.
%
If there exists a set $\P$ of $n$ points satisfying $\C$ and contained in a  hyperplane, then there exists a set of $n$ points $\P'$ contained in a hyperplane and a linear extension $\C'$ of $\C$ satisfied by $\P'$.
\end{lemma}

\begin{proof}
Assume that, for every row of $M(\P)$ except the first,
all entries in this row are distinct. Then the linear orderings of these real values in each row define a linear ordering configuration $\C'$. This configuration $\C'$ is a linear extension of $\C$ since $\P$ satisfies $\C$. Then $\P'=\P$ and $\C'$ have the required properties.

Assume two columns with labels $e,f\in \E$ of $M(\P)$ are equal.
% corresponding to the points labelled by $e,f\in \E$. 
Let $v=(v_1,\ldots,v_{n-1})$ be a vector parallel to the (affine) hyperplane containing $\P$. Let $\P'$ be obtained by 
adding $\varepsilon.(0,v_1,\ldots,v_{n-1})$ to coordinates of the point labelled by $f$, for some $\varepsilon>0$. 
%replacing $f$ with $f+\varepsilon.v$ for some $\varepsilon>0$. 
Obviously, the value $\varepsilon$ can be chosen small enough in order to have that $\P'$ still satisfies $\C$. By definition of $v$, $\P'$ is contained in the same hyperplane as $\P$.
Iteratively using this construction ultimately yields a set of points $\P'$ satisfying $\C$, contained in the same hyperplane as $\P$ and such that the columns of $M(\P')$ are all distinct.

Assume that in the row $b\in\B$ in $M(\P)$, the value in columns labelled by $e,f\in \E$ are the same. Up to transforming $\P$ as above, we assume that those two columns are not equal. Then, there exists a row $b'\in\B$ such that $x_{e,b'}\not=x_{f,b'}$. 
Let $\P'$ be the set of points whose matrix $M(\P')$ is obtained by adding $\varepsilon$ times row $b'$ to row $b$ in $M(\P)$, for some $\varepsilon>0$.
Obviously, the value $\varepsilon$ can be chosen small enough to have that $\P'$ still satisfies $\C$. Since the determinant of the matrices $M(\P)$ and $M(\P')$ are equal, $\P'$ is contained in a hyperplane. 
%(not necessariy the same as $\P$).
Using this construction iteratively  ultimately yields a set of points $\P'$ satisfying $\C$, contained in a hyperplane and satisfying the hypothesis presented in the first paragraph of this proof.
\end{proof}

\begin{proposition}
%\label{cas_particulier}
\label{prop:extension}
Let $\mathcal{C}$ be a configuration on $(\E,\B)$. 
%The configuration $\mathcal{C}$ is non-fixed if and only if there exists a non-fixed linear extension of $\mathcal{C}$.
The configuration $\mathcal{C}$ is non-fixed if and only if 
%either there exist two elements of $\E$ comparable in no ordering in $\C$, or 
there exists a non-fixed linear extension of $\mathcal{C}$.
The configuration $\mathcal{C}$ is fixed if and only if 
every linear extension of $\mathcal{C}$ is fixed.
\end{proposition}

%\conf
%From the two previous lemmas, we get the above proposition. 
%Thanks to this result, in the rest of the paper we can restrict the problem to the case of linear ordering configurations.

\begin{proof}
We prove the first assertion in the proposition. 
The second is obviously equivalent.
By Lemma \ref{lem:non-fixed}, %Observation \ref{obs:non-fixed}, 
if $\C$ is non-fixed then there exists a set $\P$ of $n$ points  satisfying $\C$ such that $det(M(P))=0$, that is $\P$ is contained in a hyperplane.
Lemma~\ref{extLin} implies that there exists $\P'$ satisfying a linear extension $\C'$ of $\C$ and such that $det(M(\P'))=0$. Hence $\C'$ is non-fixed.
Conversely, let $\C'$ be a non-fixed linear extension of $\C$.
By Lemma \ref{lem:non-fixed}, there exists $\P$ such that $det(M(P))=0$ and $\P$ satisfies $\C'$. In particular, $\P$ satisfies $\C$ , and hence $\C$ is non-fixed.
\end{proof}

With the above result, we only need to test the fixity of linear ordering configurations in order to deduce the fixity of any configuration. 
In the following, we will concentrate on linear ordering configurations. 
%Also, recall that permutations of points and coordinate do not affect the fixity property. 

%%%%%%%%%%%%%%%%%%%%%%%%%%%%%%%%%%%%%%%%%%%%%%

\subsection{Formal fixity}
%\subsection{Algebraic fixity}

Let $\C$ be a linear ordering configuration on $(\E,\B)$. 
We consider formal expressions of type $x_{e,b}-x_{f,b}$ for $e,f\in\E$, $e\not=f$, and $b\in\B$,
which we may sometimes denote $x_{e-f,b}$ for short.
Such a formal expression gets a \emph{formal sign w.r.t. $\C$} 
%- or \emph{sign} for short - 
denoted $\sigma_\C(x_{e,b}-x_{f,b})$ and belonging to $\{\plus,\moins\}$, the following way: 
%Each formal expression of type $x_{e,b}-x_{f,b}$ for $e,f\in\E$, $e\not=f$, and $b\in\B$,
%which we may sometimes denote $x_{e-f,b}$ for short, gets a \emph{formal sign w.r.t. $\C$}, or \emph{sign} for short, belonging to $\{\plus,\moins\}$, the following way: 
%$\plus$ if $f <_b e$ and  $\moins$ if $e <_b f$.
\begin{displaymath}
\begin{array}{c c c}
\sigma_\C(x_{e,b}-x_{f,b})\ =\ \ 
\plus \ \ \hbox{ if }\ \ f <_b e; 
&\hskip 2cm &
\sigma_\C(x_{e,b}-x_{f,b})\ =\ \ 
\moins\ \  \hbox{ if } \ \ e <_b f. 
\end{array}
\end{displaymath}
Recall that the polynomial $det(M_{\E,\B})$ is a multivariate polynomial on variables $x_{e,b}$ for $b\in\B$ and $e\in\E$.
%Assume $det(M_{\E,\B})$ is given by an expression equal to that of a multivariate polynomial where each variable is replaced by some $x_{e,b}-x_{f,b}$, for $b\in\B$ and $e,f\in\E$.
Assume a particular formal expression of $det(M_{\E,\B})$ is a sum of multivariate monomials where each variable is replaced by some $x_{e,b}-x_{f,b}$, for $b\in\B$ and $e,f\in\E$.
Various expressions of this type can be obtained by suitable transformations and determinant cofactor expansions from the matrix $M$, as we will do more precisely below.
This particular expression of $det(M_{\E,\B})$ gets a \emph{formal sign w.r.t. $\C$} 
%- or \emph{sign} for short - 
belonging to $\{\plus,\moins, \indet\}$, by replacing each expression of type $x_{e,b}-x_{f,b}$ with its formal sign $\sigma_\C(x_{e,b}-x_{f,b})$ 
and applying the following formal calculus rules:
\begin{displaymath}
\begin{array}{c c c}
\plus \cdot \plus = \moins \cdot \moins = \plus, &\hskip 2cm &\plus \cdot \moins = \moins \cdot \plus = \moins, 
\end{array}
\end{displaymath}
%%\begin{displaymath}
%%\begin{array}{c c}
%%+ \cdot \pm = \indet&- \cdot \pm = \indet
%%\end{array}
%%\end{displaymath}
%%\smallskip
%\begin{displaymath}
%\begin{array}{c c c c c}
%\plus + \plus = \plus - \moins =\plus,&\hskip 1cm &\moins + \moins = \moins - \plus = \moins,&\hskip 1cm &\plus + \moins =\moins + \plus= \indet,
%\end{array}
%\end{displaymath}
%%\begin{displaymath}
%%\begin{array}{c c}
%%\plus + \indet = \indet&\moins + \indet = \indet
%%\end{array}
%%\end{displaymath}
\begin{displaymath}
\begin{array}{c c c}
\plus + \plus = \plus - \moins =\plus,&\hskip 1cm &\moins + \moins = \moins - \plus = \moins,
\end{array}
\end{displaymath}
\begin{displaymath}
\begin{array}{c}
\plus + \moins =\moins + \plus= \indet,
\end{array}
\end{displaymath}
and the result of any operation involving a $\indet$ term or factor is also $\indet$.
\smallskip

We say that $\C$ is \emph{formally fixed}
 if $det(M_{\E,\B})$ has such a formal expression whose formal sign is not $\indet$.
% Obviously, if $\C$ is algebraically fixed, then the evaluation of the determinant for any set of real values $\P$ satisfying $\C$ provides a real number whose sign is consistent with the formal sign of this expression. In this case, this resulting sign does not depend on the chosen expression as soon as it is not $\indet$, and $det_C(M)$ equals this sign, and hence $\C$ is fixed.
 
 \bigskip
\noindent{\it Example.}
Consider the following matrix $M=M_{\E,\B}$ for $\E=\{a,b,c\}$ and $\B=\{1,2\}$:

\begin{minipage}[c]{1\linewidth}
\begin{displaymath}
M = \left(\begin{array}{c c c}
1&1&1\\
x_{a,1}&x_{b,1}&x_{c,1}\\
x_{a,2}&x_{b,2}&x_{c,2}
\end{array}\right)
\end{displaymath}
\end{minipage}
and consider the configuration $\C$ defined by:

\begin{minipage}[c]{1\linewidth}
\begin{displaymath}
\begin{array}{c c c c c}
a &<_1& b &<_1& c \\
b &<_2& c &<_2& a \\
\end{array}
\end{displaymath}
\end{minipage}
A formal expression of $det(M)$ is:

\begin{minipage}[c]{1\linewidth}
\begin{displaymath}
det(M)=
x_{b-a,1}\cdot x_{c-a,2} -  x_{b-a,2}\cdot x_{c-a,1}
\end{displaymath}
\end{minipage}
whose formal sign w.r.t. $\C$ is

\begin{minipage}[c]{1\linewidth}
\begin{displaymath}
\plus\cdot\moins - \moins\cdot\plus\ \ =\ \ \indet.
\end{displaymath}
\end{minipage}
Another formal expression of $det(M)$ is:

\begin{minipage}[c]{1\linewidth}
\begin{displaymath}
det(M)=
x_{b-a,1}\cdot x_{c-b,2} -  x_{b-a,2}\cdot x_{c-b,1}
\end{displaymath}
\end{minipage}
whose formal sign w.r.t. $\C$ is

\begin{minipage}[c]{1\linewidth}
\begin{displaymath}
\plus\cdot\plus - \moins\cdot\plus\ \ =\ \ \plus.
\end{displaymath}
\end{minipage}
This second expression shows that $\C$ is formally fixed.
\bigskip

 \begin{observation}
 \label{obs:alg-fixed}
If $\C$ is formally fixed, then $\C$ is fixed.
 \end{observation}

More precisely, given an expression (as above) whose formal sign w.r.t. $\C$ is $\plus$ or $\moins$,
the evaluation of this determinant for any set of real values $\P$ satisfying $\C$ necessarily provides a real number whose sign is consistent with the formal sign of this expression. In this case, this resulting sign does not depend on the chosen expression, as long as it is not $\indet$, and $\sigma_C(det(M))$ %$det_C(M)$
equals this sign.
\medskip

%Conversely, it is not obvious that for every fixed configuration there would exist a suitable expression of the determinant showing formally that $\C$ is fixed by this way. However, we strongly believe in this result, which we state as a conjecture, and which we will prove for $n\leq 4$.

Conversely, 
one may wonder if
%it is not obvious that 
for every fixed configuration $\C$ there would exist a suitable expression of the determinant formally showing  in the above way that $\C$ is fixed.
%\corr{ by this way} in the same manner\AFINIR(de cette maniere). 
That is, equivalently, do we have: if every formal expression of $det(M_{\E,\B})$ has formal sign $\indet$, then 
$\sigma_C(det(M))=\spm\ $?
%However, 
We strongly believe in this result, which we state as a conjecture, and which we will prove for $n\leq 4$ (see Theorems \ref{fixe2D} and \ref{caract}).

\begin{conjecture}
\label{conjecture1}
Let $\C$ be a linear ordering configuration on $(\E,\B)$. Then $\C$ is fixed if and only if $\C$ is formally fixed.
\end{conjecture}

%We conjecture this converse implication, that is the equivalence between being fixed and formally fixed,
%and even the equivalence with a more constrained notion for formally fixed that we detail below (see Conjecture \ref{conjecture}). % (see Section \ref{sec:dim-n}).
%%%We will prove these equivalences in dimensions 2 and 3, together with more precise characterizations in those dimensions (see Sections \ref{sec:dim-2} and \ref{sec:dim-3}). 
%\bigskip

%%%%%%%%%%%%%%%%%%%%%%%%%%%%%%%%%%%%%%%%%%%%%%%%%%%%%%%%ù
\subsection{Formal fixity by expansion}

Let $\C$ be a configuration on $(\E,\B)$, and $\E'=\E\setminus\{e\}$, $\B'=\B\setminus\{b\}$ for some $e\in\E$, $b\in\B$. 
We call \emph{configuration induced by $\C$ on $(\E',\B')$} the configuration on $(\E',\B')$  obtained by restricting every ordering $<_{b'}$, $b'\in \B'$, of $\C$ to $\E'$.
% Obviously, the matrix $M_{\E',\B'}$ is obtained from $M_{\E,\B}$ by deleting the row corresponding to $b$ and the column corresponding to $e$. This matrix is called \emph{restriction of $M_{\E,\B}$ to $(\E',\B')$}.
 Moreover, we say that \emph{all the configurations induced by $\C$ on $\mathcal{E}'$ are fixed} if, for every  $b\in \B$,
 %for every ordering $<_b$, $b\in \B$, 
 the configuration induced by $\C$ on $({\E}',\B\setminus\{b\})$ is a fixed configuration.
%\bigskip
Note that, from a geometrical viewpoint, if $\P$ is a set of points satisfying $\C$, and $\P_e$ is obtained by removing the point with label $e\in\E$  from $\P$, then the projection $\P'$ of $\P_e$ on $\B'$ along $b$ satisfies $\C'$. Indeed, the matrix $M_{\E',\B'}$, resp. $M_{\E',\B'}(\P')$, is obtained by removing
the column corresponding to $e$ and the row corresponding to $b$ from $M_{\E,\B}$, resp. $M_{\E,\B}(\P)$.
\smallskip

As previously, let  $M=M_{\E,\B}$ with $\E=\{e_1,...,e_{n}\}_<$ and $\B=\{b_1,...,b_{n-1}\}_<$.
%For $e\in\E$ we denote $C_e$ the column of $M$ correponding to $e$. 
Let $e_i,e_j\in\E$, with $e_i\not=e_j$. 
Consider the matrix obtained from $M$ by subtracting the $j$-th column (corresponding to $e_j$),  from the $i$-th column (corresponding to $e_i$), that is:
\begin{displaymath}
%\bar M=
\left(\begin{array}{c c c c c c c}
1&\ldots&1&0&1&\ldots&1\\
x_{e_1,b_1}&\ldots&x_{e_{i-1},b_1}&x_{e_i,b_1}-x_{e_j,b_1}&x_{e_{i+1},b_1}&\ldots&x_{e_n,b_1}\\
x_{e_1,b_2}&\ldots&x_{e_{i-1},b_2}&x_{e_i,b_2}-x_{e_j,b_2}&x_{e_{i+1},b_2}&\ldots&x_{e_n,b_2}\\
\vdots&&\vdots&\vdots&\vdots&&\vdots\\
x_{e_1,b_{n-1}}&\ldots&x_{e_{i-1},b_{n-1}}&x_{e_i,b_{n-1}}-x_{e_j,b_{n-1}}&x_{e_{i+1},b_{n-1}}&\ldots&x_{e_n,b_{n-1}}
\end{array}\right)
\end{displaymath}

The determinant of this matrix equals $det(M)$. 
The cofactor expansion formula for the determinant of this matrix w.r.t. its $i$-th column yields:
$$det(M_{\E,\B})\ =\ \sum_{k=1}^{n-1}\ (-1)^{i+k+1}\ \cdot\ (x_{e_i,b_k}-x_{e_j,b_k})\ \cdot\ det\bigl(\ M_{\E\setminus\{e_i\},\B\setminus\{b_k\}}\ \bigr)
$$
which we call \emph{expression of $det(M)$ by expansion with respect to $(e_i,e_j)$}.
\bigskip

%*** notations supprimees: operation de colonnes, restriction to E',B'***
%
%Now, consider an expression of $det(M)$ by expansion w.r.t. $(e_i,e_j)$ as above and let $\C_k$ be the projection of $\C$ w.r.t. $b_k$ restricted to $\E\setminus\{e_i\}$.
%Let $\C_k$ be the projection of $\C$ w.r.t. $b_k$ restricted to $\E\setminus\{e_i\}$.
Then the above particular expression of $det(M)$ 
gets a \emph{formal sign} w.r.t. $\C$ in the following manner.
First, replace 
each expression of type $x_{e,b}-x_{f,b}$ with its formal sign w.r.t. $\C$ in $\{\plus,\moins\}$,
 and replace each $det(M_{\E\setminus\{e_i\},\B\setminus\{b_k\}})$, $1\leq k\leq n-1$,
with its  sign  %$det_{\C_k}(M_{\E\setminus\{e_i\},\B\setminus\{b_k\}})\in\{\plus,\moins\spm\}$,
$\sigma_{\C_k}(det(M_{\E\setminus\{e_i\},\B\setminus\{b_k\}}))\in\{\plus,\moins,\spm\}$,
where $\C_k$ is the configuration induced by $\C$ on 
$(\E\setminus\{e_i\},\B\setminus\{b_k\})$.
%w.r.t. $b_k$ restricted to $\E\setminus\{e_i\}$.
 This leads to the formal expression:
%$$\sum_{k=1}^{n-1}\ (-1)^{i+k+1}\ \cdot\ \sigma_\C(x_{e_i,b_k}-x_{e_j,b_k})\ \cdot\ det_{\C_k}\bigl(\ M_{\E\setminus\{e_i\},\B\setminus\{b_k\}}\ \bigr),
%$$
$$\sum_{k=1}^{n-1}\ (-1)^{i+k+1}\ \cdot\ \sigma_\C(x_{e_i,b_k}-x_{e_j,b_k})\ \cdot\ \sigma_{C_k}\Bigl(det\bigl(\ M_{\E\setminus\{e_i\},\B\setminus\{b_k\}}\ \bigr)\Bigr),
$$
%performing the substitution 
%%$$ |\cdot|_{i,j}\ \ = 
%$$ det(M''_{i,j})\ \ = 
% \ \ (-1)^{\sigma_i}\ det_{\C_j}(M_{\E\setminus\{e\},\B\setminus\{b_j\}})$$
Then, provide the formal sign of this expression by using the same formal calculus rules as previously, completed with the following one:
$$\plus \cdot \spm\ \ =\ \ \moins \cdot \spm\ \ =\ \ \indet.$$
%\begin{displaymath}
%\begin{array}{c c c}
%\plus \cdot \spm = \indet,& \hskip 2cm&\moins \cdot \spm = \indet.
%\end{array}
%\end{displaymath}

%As previously, it is obvious that if the resulting formal sign is not $\indet$, then $det_C(M)$ equals this sign and the configuration $\C$ is fixed.
%In this case, that is precisely if there exists a developable matrix and developable column such that the previous expression has a determined sign $\plus$ or $\moins$, then we say that $\C$ is \emph{formally fixed by projections}.
%Notice that it implies that every involved projection is a fixed configuration.

\bigskip
%If there exists such an expression of $det(M)$ whose formal sign is $\plus$ or $\moins$, 
If there exists such an expression of $det(M)$ by expansion whose formal sign is $\plus$ or $\moins$, 
then $\C$ is called \emph{formally fixed by expansion}.
%Notice that it implies that all the projections of $\C$ restricted to $\E'$ are fixed. 

%\bigskip

 \begin{observation}
 \label{obs:alg-fixed-proj}
If $\C$ is formally fixed by expansion, then $\C$ is fixed.
 \end{observation}

The above observation is similar to Observation \ref{obs:alg-fixed}:
if  $\C$ is formally fixed by expansion then $\sigma_\C(det(M))$ is given as the formal sign of any expression certifying that $\C$ is formally fixed by expansion.
Notice that if $\C$ is formally fixed by expansion then all the configurations $\C_k$ induced by $\C$  are fixed, since we must have %$det_{\C_k}(M_{\E\setminus\{e_i\},\B\setminus\{b_k\}})\in\{\plus,\moins\}$.
$\sigma_{\C_k}(det(M_{\E\setminus\{e_i\},\B\setminus\{b_k\}}))\in\{\plus,\moins\}$.

\begin{conjecture}
\label{conjecture2}
Let $\C$ be a linear ordering configuration on $(\E,\B)$. Then $\C$ is fixed if and only if $\C$ is formally fixed by expansion.
\end{conjecture}

We point out that if Conjecture \ref{conjecture1} is true in dimension $n-1$, then Conjecture \ref{conjecture2} in dimension $n$ implies Conjecture \ref{conjecture1} in dimension $n$.
Indeed, in this case, the fixity of the $(n-1)$-dimensional configurations corresponding to cofactors can be determined using formal expressions. 
\medskip

%\noindent{\it Remark.}
\begin{remark}
\label{rk:SNS}
Assume $n=4$, consider any  $3\times 3$ matrix $M'$ obtained from $M$ by subtracting some columns, and deleting the first row and one column, so that every entry in $M'$ is of type $x_{e,b}-x_{f,b}$ for $e,f\in\E$ and $b\in\B$. We have  either $det(M')=det(M)$ or $det(M')=-det(M)$. Then replace  in the matrix $M'$ each formal expression $x_{e,b}-x_{f,b}$ with its formal sign $\sigma_\C(x_{e,b}-x_{f,b})$ w.r.t. to a given configuration $\C$. 
We obtain a $3\times 3$ matrix $N$ with entries in $\{\plus,\moins\}$.
The point of this remark is that formally computing the sign of the determinant of the matrix $N$, using the same formal rules as above, will always provide the result $\indet$. The proof of this property is left as an exercise to the reader. 
In fact, 
as already noticed in the introduction of the paper,
this property generalizes in any dimension, it is known as:
an SNS-matrix of order $n\geq 3$ has at least one
zero, see \cite[page 108]{SNS}.
This shows that a  formal matrix $M'$, such as the above one, cannot be used alone to derive a formal expression of the determinant of the original matrix $M$ proving the fixity of a configuration. One would always need to transform submatrices of $M$, which is what we do implicitly by the inductive use of formal signs of induced configurations in order to determine formal fixity.
\end{remark}
\medskip
%
%\begin{conjecture}
%\label{conjecture}
%Let $\C$ be a linear ordering configuration on $(\E,\B)$. The following properties are equivalent:
%
%(a) $\C$ is fixed,
%
%(b) $\C$ is formally fixed
%
%(c) $\C$ is formally fixed by projections
%\end{conjecture}
%
%In what follows, we will notably prove this conjecture in dimensions 2 and 3.
%We will prove these equivalences in dimensions 2 and 3, together with more precise characterizations in those dimensions (see Sections \ref{sec:dim-2} and \ref{sec:dim-3}). 

%Finally, the point of this paper is to deal with the property of being formally fixed by expansion as an inductive criterion for fixity.
%From the algorithmic viewpoint, this consruction can be seen as an inductive way to build suitable formal expressions for $det(M)$.
%In what follows, we will prove the above conjecture for $n=4$,
%together with more precise and direct characterizations in this case.

Finally, the point of this paper is to deal with the property of being formally fixed by expansion as an inductive criterion for fixity.
%:
%suitable cofactor expansions are used to build suitable formal expressions for $det(M)$.
%From the algorithmic viewpoint, this construction can be seen as an inductive way to build suitable formal expressions for $det(M)$.
Next, we will prove Conjecture \ref{conjecture2} for $n=4$,
providing at the same time more precise and direct characterizations in this case (see Theorem \ref{caract}).

%%%%%%%%%%%%%%%%%%%%%%%%%%%%%%%%%%%%%%%%%%%%%%%%%%%%%%%%%%%%%%%%%%%%%%%%%%%%%ù
%\subsection{montrer qu'une configuration est non-fixe}
%Nous allons maintenant introduire un lemme permettant de montrer qu'une configuration est non-fixe.

\subsection{A non-fixity criterion}

The following Lemma \ref{non-fixe} will be our main tool to prove that a configuration is non-fixed.
We point out that, when $n=4$, the sufficient condition for being non-fixed provided by Lemma \ref{non-fixe}  turns out to be a  necessary and sufficient condition (see Theorem \ref{caract-non-fixed}).
However, the authors feel that this equivalence result is 
too hazardous %than Conjectures \ref{conjecture1} and \ref{conjecture2} 
to be stated as a general conjecture in dimension $n$.
%We mention that it generalizes directly in any dimension $n$, with a similar proof.
%We say that an element of $\E$ is \emph{extreme} in a linear ordering on $\E$ if it is the first or last element in this ordering.

%*** A DECIDER : l'enonce t'on aussi comme conjecture ? ***

%\begin{lemma}
%\label{non-fixe}
%Let ${\C}$ be a configuration on $(\E,\B)$.
%If the configuration ${\C}'$ induced by ${\C}$ on $(\E\setminus \{e\},\B\setminus\{b\})$ for some $e\in\E$ and $b\in\B$ satisfies the following properties:
%%If a projection ${\C}'$ of ${\C}$ w.r.t. $i\in\B$ restricted to $\E\setminus \{e\}$ satisfies the following properties:
%\begin{itemize}
%	\item[\textbullet] ${\C}'$ is non-fixed and
%	\item[\textbullet] $e$ is extreme in the ordering $<_b$ of ${\C}$,
%\end{itemize}
%then ${\C}$ is non-fixed.
%\end{lemma}

\begin{lemma}
\label{non-fixe}
Let ${\C}$ be a configuration on $(\E,\B)$.
If there exist $e\in\E$ and $b\in\B$ satisfying the following properties:
%If a projection ${\C}'$ of ${\C}$ w.r.t. $i\in\B$ restricted to $\E\setminus \{e\}$ satisfies the following properties:
\begin{itemize}
	\item[\textbullet] $e$ is extreme in the ordering $<_b$ of ${\C}$ and
	\item[\textbullet] the configuration ${\C}'$ induced by ${\C}$ on $(\E\setminus \{e\},\B\setminus\{b\})$ is non-fixed,
\end{itemize}
then ${\C}$ is non-fixed.
\end{lemma}

\begin{proof}
To lighten notations, let us denote $\E=\{1,...,n\}$ and $\B=\{1,...,n-1\}$. Up to equivalence of configurations, we can assume that $e=1$, that $b=1$ and that $1$ is minimal in the ordering $<_1$. 

The expression of $det(M)$ by expansion with respect to $(1,2)$ yields:

\begin{minipage}[c]{1\linewidth}
\begin{eqnarray*}
det(M_{\E,\B})\ &=\ &\sum_{k=1}^{n-1}\ (-1)^{k}\ \cdot\ (x_{1,k}-x_{2,k})\ \cdot\ det\bigl(\ M_{\E\setminus\{1\},\B\setminus\{k\}}\ \bigr)
\\
  &=\ &(x_{2,1}-x_{1,1})\cdot det(M_{\E\setminus\{1\},\B\setminus\{1\}})
%+P[x_{i,j}]_{1\leq i\leq n,\ 1\leq j\leq n-1,\ (i,j)\not=(1,1)}.
  +P[x_{i,j}]_{(i,j)\not=(1,1)}.
\\
\end{eqnarray*}
\end{minipage}
where 
%$P[x_{i,j}]_{1\leq i\leq n,\ 1\leq j\leq n-1,\ (i,j)\not=(1,1)}$ 
$P[x_{i,j}]_{(i,j)\not=(1,1)}$
is a polynomial in the same variables as $M_{\E,\B}$ not depending on $x_{1,1}$. %(we will denote simply $P$ this polynomial in what follows).

By hypothesis, the configuration $\C'$ 
%of $C$ w.r.t. $1$ and restricted to $\E\setminus\{1\}$ 
is non-fixed, that is $\sigma_{\C'}(det(M_{\E\setminus\{1\},\B\setminus\{1\}}))=\spm$. 
By Lemma \ref{lem:non-fixed},
there exist real values $\P'_+$ and $\P'_-$ for the entries of this matrix,
that is two sets of $n-1$ points labeled by $\E\setminus \{1\}$ in dimension $n-2$,
such that $det(M_{\E\setminus\{1\},\B\setminus\{1\}}(\P'_+))>0$ and  $det(M_{\E\setminus\{1\},\B\setminus\{1\}}(\P'_-))<0$.

Let us define a set of $n$ points $\P_+$ labeled by $\E$  in dimension $n-1$  the following manner.
The formal variables in $M_{\E,\B}$ with real values specified by $\P'_+$ 
get the same values in $\P$.
All values not specified by $\P'_+$ except $x_{1,1}$ are  fixed arbitrarily but consistently with the orderings in $\C$.
The value $x_{1,1}$ is chosen small enough so that $x_{1,1}$ is minimal in $<_1$ and
$$(x_{2,1}-x_{1,1})\cdot det(M_{\E\setminus\{1\},\B\setminus\{1\}}(\P'_+))>- P[x_{i,j}]_{(i,j)\not=(1,1)}$$
This is possible since $det(M_{\E\setminus\{1\},\B\setminus\{1\}}(\P'_+))>0$ and the second term of the inequality does not depend on $x_{1,1}$. 
By this definition, we have obtained $det(M(\P_+))>0$.

Similarly, we define $\P_-$ by choosing $x_{1,1}$ small enough so that $x_{1,1}$ is minimal in $<_1$ and
$$(x_{2,1}-x_{1,1})\cdot det(M_{\E\setminus\{1\},\B\setminus\{1\}}(\P'_-))<- P[x_{i,j}]_{(i,j)\not=(1,1)}$$
This is possible since $det(M_{\E\setminus\{1\},\B\setminus\{1\}}(\P'_-))<0$ and the second term of the inequality does not depend on $x_{1,1}$. 
By this definition, we have obtained $det(M(\P_-))<0$.

We have built $\P_+$ and $\P_-$ providing opposite signs to real evaluations of $det(M_{\E,\B})$. That is, by Lemma \ref{lem:non-fixed}, %Observation \ref{obs:non-fixed}, 
$\C$ is non-fixed.
\end{proof}

%%%%%%%%%%%%%%%%%%%%%%%%%%%%%%%%%%%%%%%%%%%%%%%%%%%ù
%%%%%%%%%%%%%%%%%%%%%%%%%%%%%%%%%%%%%%%%%%%%%%%%%%%%%ù
%%%%%%%%%%%%%%%%%%%%%%%%%%%%%%%%%%%%%%%%%%%%%%%%%%%

\section{Characterizations in low dimensions}

%%%%%%%%%%%%%%%%%%%%%%%%%%%%%%%%%%%%%%%%%%%%%%%%%%%%%%%%%%%%%%%%%%%%%%%%
%\section{Results in dimension 2}
\subsection{Results in dimension 2}
\label{sec:dim-2}

%Pour alléger les notations nous noterons $\mathcal{P}=\{A,B,C\}$ l'ensemble de 3 points et $x,y$ les axes. La matrice $M_{\mathcal{P}}$ est donc la suivante : 
In this section we fix $n=3$ and $\E=\{A,B,C\}$.
In order to lighten notations of variables $x_{e,b}$ for $e\in\E$ and $b\in \B$, we sooner denote:
\begin{displaymath}
M=\left(\begin{array}{c c c}
1&1&1\\
x_A&x_B&x_C\\
y_A&y_B&y_C\\
\end{array}\right)
\end{displaymath}
We also denote $\B=\{x,y\}$ and $<_x,<_y$ the orderings in a configuration.

%\subsection{Caractérisation des configurations fixes et des configurations non-fixes}
%D'après ce que nous venons de voir dans la partie précédente, il nous suffit de nous intéresser aux cas où les ordres sont totaux. Nous cherchons donc à placer les points $A$, $B$, $C$ dans un plan sans que 2 coordonnées soient égales. Il n'y a que 2 façons de positionner nos 3 points en dimension 2 à rotations, symétries et renumérotation des points près : 

It is easy to verify that, up to equivalence of configurations, 
%(i.e. up to rotations, symmetries and relabelling), 
there exist exactly two linear ordering configurations:

\begin{minipage}[b]{0.45\linewidth}
\begin{displaymath}
\begin{array}{c}
%x_A < x_B < x_C \\
%y_A < y_B < y_C
A <_x B <_x C \\
A <_y B <_y C
\end{array}
\end{displaymath}
\end{minipage}
\begin{minipage}[b]{0.45\linewidth}
\begin{displaymath}
\begin{array}{c}
%x_A < x_B < x_C \\
%y_B < y_C < y_A
A <_x B <_x C \\
B <_y C <_y A
\end{array}
\end{displaymath}
\end{minipage}

\smallskip
\noindent 
which correspond to the following respective grid representations:

\begin{minipage}[b]{0.45\linewidth}
\centering \includegraphics[width=2.5cm]{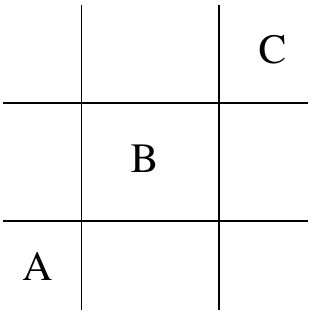}
\end{minipage}
\begin{minipage}[b]{0.45\linewidth}
\centering \includegraphics[width=2.5cm]{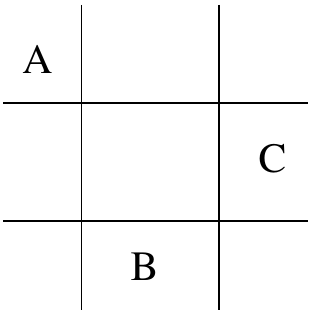}
\end{minipage}

%En 2 dimensions, nous avons un critère pour caractériser les configurations fixes et les configurations non-fixe :

\begin{theoreme}
\label{non-fixe2D}
Let ${\C}$ be a linear ordering configuration on $(\E,\B)$ with $n=3$, $\E=\{A,B,C\}$ and $\B=\{x,y\}$.  
The following properties are equivalent:

a) ${\C}$ is non-fixed;

b) the two orderings on $\E$ in $\C$ are either equal or equal to reversions of each other;

c) 
up to equivalence,
%up to permutations of $\E$ and $\B$, and up to inversions of orderings in $\C$, 
${\C}$ is equal to 
\begin{displaymath}
\begin{array}{c}
%x_A < x_B < x_C \\
%y_A < y_B < y_C
A <_x B <_x C \\
A <_y B <_y C
\end{array}
\end{displaymath}
\end{theoreme}

\begin{proof}
The equivalence between b) and c) is straightforward and left to the reader.
%On laissera au soin du lecteur le fait de vérifier que b) et c) sont équivalent. 
%
Let us prove that c) implies a). Let $\C$ be given by condition c).
Let us choose $\P$ satisfying $\C$ and $x_A=y_A$,  $x_B=y_B$, $x_C=y_C$.
We have $det(M(\P))=0$, hence $\C$ is non-fixed
by Lemma \ref{lem:non-fixed}. %Observation \ref{obs:non-fixed}.
%Montrons que c) implique a). D'après le Lemme~\ref{hyperplanNonFixe}, pour montrer que ${\C}_T$ est non-fixe, il suffit de trouver 3 points vérifiant ${\C}_T$ et appartenant à un même hyperplan. Or les hyperplans en dimension 2 correspondent aux droites. En prenant $\mathcal{P}=\{A,B,C\}$ l'ensemble de point pour lequel $A=(1,1)$, $B=(2,2)$ et $C=(3,3)$, $\mathcal{P}$ respecte bien ${\C}_T$ et les points de $\mathcal{P}$ appartiennent à une même droite. Ainsi ${\C}_T$ est bien non-fixe.
%
In order to prove that a) implies c), we can equally prove that the  other possible linear ordering configuration (up to equivalence) is fixed. This result is given by Theorem \ref{fixe2D} below.
%Pour montrer que a) implique c) il nous suffit de montrer que l'autre configuration d'ordres totaux est fixe. C'est ce que nous allons faire dans le théorème suivante.
\end{proof}

%***IDEE: ajouter "equivalence" de deux configurations pour up to symmetries, rotations and relabelling dans la section preliminaire***

%A l'inverse, le critère pour caractériser les configurations fixes est le suivant :
%\begin{theoreme}
%\label{fixe2D}
%Let ${\C}_T$ be a configuration of two total orders. Then a), b) and c) are equivalent :
%
%a) ${\C}_T$ is fixed ;
%
%b) except permutations of ordres, inversions of orders and renumbering of points, ${\C}_T$ is equal to
%\begin{displaymath}
%\begin{array}{c c c}
%x_A < x_B < x_C \\
%y_B < y_C < y_A
%\end{array}
%\end{displaymath}
%
%c) % EN TRAVAUX
% there exists a developable matrix of the configuration ${\C}_T$ (denote $M_1$) and a developing column $j$ such that one algebraic development of $\det(M_1)$ by $j$ on the configuration ${\C}_T$ have this sign in $\{\plus,\moins\}$.
%\end{theoreme}

%In the theorem below we write formally fixed (by projections) eaning that the brackets can be read or not while not changing the correctness of the result.
%
%*** ou alors ajouter une equivalence pour bien dissocier ***

\begin{theoreme}
\label{fixe2D}
Let ${\C}$ be a linear ordering configuration on $(\E,\B)$ with $n=3$, 
$\E=\{A,B,C\}$ and $\B=\{x,y\}$. 
The following properties are equivalent:

a) ${\C}$ is fixed;

%b) $\C$ is formally fixed (by projections);

b) $\C$ is formally fixed;

c) up to equivalence,
%except permutations of ordres, inversions of orders and renumbering of points, 
${\C}$ is equal to
\begin{displaymath}
\begin{array}{c c c}
%x_A < x_B < x_C \\
%y_B < y_C < y_A
A <_x B <_x C \\
B <_y C <_y A
\end{array}
\end{displaymath}
%
%c) % EN TRAVAUX
% there exists a developable matrix of the configuration ${\C}_T$ (denote $M_1$) and a developing column $j$ such that one algebraic development of $\det(M_1)$ by $j$ on the configuration ${\C}_T$ have this sign in $\{\plus,\moins\}$.
\end{theoreme}

\begin{proof}
%Montrons que b) implique c). On peut transformer la matrice $M_{\mathcal{P}}$ de la façon suivante :
Recall that b) implies a) is always true.
Let us prove that c) implies b). 
We have:

\begin{minipage}[c]{1\linewidth}
\begin{eqnarray*}
det(M) & =
& det \left(\begin{array}{c c c}
0&1&0\\
x_{A-B}&x_B&x_{C-B}\\
y_{A-B}&y_B&y_{C-B}
\end{array}\right)\\
 & = &(x_A-x_B)\cdot(y_C-y_B)-(y_A-y_B)\cdot(x_C-x_B).\\
\end{eqnarray*}
\end{minipage}
The formal sign of this expression of $det(M)$ w.r.t. $\C$ is $$\moins\cdot\plus - \plus\cdot\plus=\moins.$$
Hence $\C$ is formally fixed.
%We can equally write that $\C$ is formally fixed (by projections) since in dimension 2 the development w.r.t. a column provides the same xepression as above.

%Montrons que c) implique a). On vient de voir que la matrice $M_1=\left(\begin{array}{c c}
%x_{A-B}&x_{C-B}\\
%y_{A-B}&y_{C-B}
%\end{array}\right)$ est une matrice développable de la configuration ${\C}_T$. De plus le signe du développement algébrique de $\det(M_1)$ dans la configuration ${\C}_T$ est égal à $\moins$. Cela signifie que pour tout ensemble de $3$ points respectant ${\C}_T$, $\det(M_1)$ est négatif. Or on a vu que $\det(M_{\mathcal{P}})=-\det(M_1)$. Du coup $\det(M)$ est positif et ce quel que soit l'ensemble de points $\mathcal{P}$ respectant ${\C}_T$. L'orientation de la base formé par $\mathcal{P}$ est donc la même pour tout ensemble de points $\mathcal{P}$ respectant ${\C}_T$. ${\C}_T$ est donc fixe.

Finally, to prove that a) implies c), we just need to prove that the other linear ordering configuration (up to equivalence) is non-fixed.
This has been already shown %as c) $\Rightarrow$ a) 
in the proof of Theorem~\ref{non-fixe2D}.
%Montrons que a) implique b). Nous avons vu dans la proposition précédente que l'autre configuration d'ordres totaux est non-fixe et que celle-ci est fixe. Ainsi si une configuration d'ordres est non-fixe alors à permutations des axes, inversions d'ordres et renumérotations des points près elle est égale à celle du c).
\end{proof}

Now that we have listed fixed and non-fixed linear ordering configurations, we are able to determine all fixed and non-fixed configurations using Proposition \ref{prop:extension}.
Let us omit configurations for which two elements of $\E$ are comparable in no ordering in the configuration, since these configurations are obviously non-fixed. Then there remain four ordering configurations which are not linear (up to equivalence of configurations), as one can easily check:

\begin{minipage}[c]{0.22\linewidth}
\begin{displaymath}
\begin{array}{c c c}
A <_x B <_x C\\
B <_y A\\
B <_y C
\end{array}
\end{displaymath}
\end{minipage}
\begin{minipage}[c]{0.25\linewidth}
\begin{displaymath}
\begin{array}{c c c}
A <_x B <_x C\\
B <_y A\\
C <_y A
\end{array}
\end{displaymath}
\end{minipage}
\begin{minipage}[c]{0.22\linewidth}
\begin{displaymath}
\begin{array}{c c c}
A <_x B <_x C
\end{array}
\end{displaymath}
\end{minipage}
\begin{minipage}[c]{0.22\linewidth}
\begin{displaymath}
\begin{array}{c c c}
A <_x C\\
B <_x C\\
B <_y A\\
C <_y A
\end{array}
\end{displaymath}
\end{minipage}

\smallskip
\begin{minipage}[t]{0.22\linewidth}
%\centering Configuration \\
%fixed
\centering fixed \\
\end{minipage}
\begin{minipage}[t]{0.25\linewidth}
\centering  non-fixed \\
%\centering {\small (because of $\scriptstyle C<_yB<_yA$)}\\
\centering {\small (because of $\scriptstyle C<_yB<_yA$, 
and implying the non-fixity of the next ones})\\
\end{minipage}
\begin{minipage}[t]{0.22\linewidth}
\centering non-fixed  \\
%\centering {\small (implied by the previous one)}\\
\end{minipage}
\begin{minipage}[t]{0.22\linewidth}
\centering non-fixed \\
%\centering {\small (implied by the same as the previous one)}\\
\end{minipage}

\bigskip
These configurations can be represented respectively in the following grids:
%Ces configurations correspondent respectivement aux figures ci-dessous.

\begin{minipage}[t]{0.22\linewidth}
\centering \includegraphics[width=2.5cm]{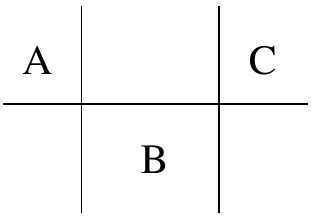}
\end{minipage}
\begin{minipage}[t]{0.25\linewidth}
\centering \includegraphics[width=2.5cm]{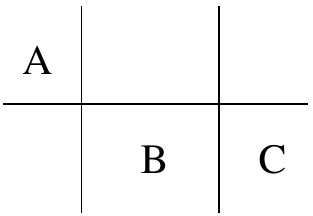}
\end{minipage}
\begin{minipage}[t]{0.22\linewidth}
\centering \includegraphics[width=2.5cm]{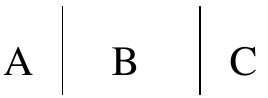}
\end{minipage}
\begin{minipage}[t]{0.22\linewidth}
\centering \includegraphics[width=2cm]{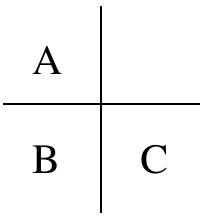}
\end{minipage}

%%%%%%%%%%%%%%%%%%%%%%%%%%%%%%%%%%%%%%%%%%%%%%%%%%%%%%%%%%%%%%%%%%%%%%%%%%%%%%%%%%%%%%%%%%%%%%%
%\section{Results in dimension 3}
\subsection{Results in dimension 3}
\label{sec:dim-3}

In this section we fix $n=4$ and $\E=\{A,B,C,D\}$.
In order to lighten notations of variables $x_{e,b}$ for $e\in\E$ and $b\in \B$, we sooner denote:
%Pour alléger les notations nous noterons $\mathcal{P}=\{A,B,C,D\}$ les points et $x,y,z$ les axes. La matrice $M_{\mathcal{P}}$ est donc la suivante : 
\begin{displaymath}
M=\left(\begin{array}{c c c c}
1&1&1&1\\
x_A&x_B&x_C&x_D\\
y_A&y_B&y_C&y_D\\
z_A&z_B&z_C&z_D
\end{array}\right)
\end{displaymath}
%and vectors will be denoted as: $v=(v_1,v_2,v_3)$.
%and vectors will be denoted as: $v=(v_x,v_y,v_z)$.
We also denote $\B=\{x,y,z\}$ and $<_x,<_y,<_z$ the orderings in a configuration.

%***NB: ai chngé $v_1$ en $v_x$ pour cohérence***
%Les vecteurs seront notés de la façon suivante : $v=(v_1,v_2,v_3)$.
\bigskip

%Dans cette partie, nous allons commencer par énoncer %les lemmes qui nous serviront pour montrer qu'une configuration est fixe ou non-fixe (Lemme~\ref{signe3ssdet} et Lemme~\ref{non-fixe}). 
%le lemme qui nous servira à montrer qu'une configuration est non-fixe. Nous allons ensuite caractériser les configurations fixes dans le Théorème~\ref{caract}. La fin de cette partie servira à prouver ce théorème.

As noted in Section \ref{sec:tools}, in order to prove that a configuration $\C$ is formally fixed by expansion, we need to find an element $e\in E$
such that all the configurations induced by $\C$ on $\E\setminus \{e\}$ are fixed. The proposition below characterizes such induced configurations.

\begin{proposition}
\label{tripletfixe}
Let ${\C}$ be a configuration on $(\E,\B)$ with $n=4$,
$\E=\{A,B,C,D\}$ and $\B=\{x,y,z\}$.
All the configurations induced by $\C$ on $\{A, B, C\}$ are fixed
%All the projections of $\C$ restricted to $\{A, B, C\}$ are fixed
if and only if 
$\C$ is equivalent to a configuration whose orderings satisfy:
\begin{displaymath}
\begin{array}{c c c c c}
B &<_x& C &<_x& A \\
C &<_y& A &<_y& B \\
A &<_z& B &<_z& C
\end{array}
\end{displaymath}
%Then, up to permutations of $\E$ and $\B$, and up to inversion of orderings in $\C$, the restriction of ${\C}_T$ to this triplet is :
%\begin{displaymath}
%\begin{array}{c c c c c}
%x_B &<& x_C &<& x_A \\
%y_C &<& y_A &<& y_B \\
%z_A &<& z_B &<& z_C
%\end{array}
%\end{displaymath}
\end{proposition}

\begin{proof}
$\Rightarrow$)
The configuration induced by $\C$ on $(\{A,B,C\},\{y,z\})$
% restricted to $\{A,B,C\}$ 
 is fixed.
According to Theorem~\ref{fixe2D}, this implies that 
the restrictions of $<_y$ and $<_z$ to $\{A,B,C\}$ are not equal nor are reversions of each other.
Similarly, the restrictions of $<_x$ and $<_y$, as well as the restriction of $<_x$ and $<_z$, to $\{A,B,C\}$ are not equal nor are reversions of each other.
These three properties easily imply the result. We omit the details.
$\Leftarrow$) The proof is direct by Theorem~\ref{fixe2D} using the same criteria as above.
\end{proof}

%%%%%%%%%%%%%%%%%%%%%%%%%%%%%%%%%%%%%%%%%%%%%%%%%%%%%%%%%%%%%%%%%%%%%%%%%%%%%%%%%%%%%%%%%%%%%

Let us now state Theorem \ref{caract}, which is the main theorem of the paper. Its proof is the content of Section \ref{sec:proof}. This proof will prove Theorem \ref{caract-non-fixed} below at the same time.

\begin{theoreme}
\label{caract}
Let ${\C}$ be a configuration on $(\E,\B)$ with $n=4$,
$\E=\{A,B,C,D\}$ and $\B=\{x,y,z\}$.
The following propositions are equivalent:

a) ${\C}$ is fixed;

b) $\C$ is formally fixed;

c) $\C$ is formally fixed by expansion;

d) up to equivalence, $\C$ satisfies:
%$\C$ is equivalent to a configuration $\C'$ satisfying:
\begin{displaymath}
\begin{array}{c c c c c}
B &<_x& C &<_x& A \\
C &<_y& A &<_y& B \\
A &<_z& B &<_z& C
\end{array}
\end{displaymath}
and there exists $X \in \{A, B, C\}$ such that 
%either $X<D$ for every  ordering $<$ in ${\C}$, or $D<X$ for every  ordering $<$ in ${\C}$.
either $X<_b D$ for every $b\in \B$, or $D<_b X$ for every $b\in \B$.

%c) First, there exists a triplet of points $\E'\subseteq \E$ such that all the projections of $\C$ restricted to $\E'$ are fixed. 
%Secondly, if we permute $\E$ and $\B$, and if we reverse some orderings in $\C$, so that  the restriction of ${\C}$ to this triplet is  
%\begin{displaymath}
%\begin{array}{c c c c c}
%x_B &<& x_C &<& x_A \\
%y_C &<& y_A &<& y_B \\
%z_A &<& z_B &<& z_C
%\end{array}
%\end{displaymath}
%(which is possible by Proposition~\ref{tripletfixe}),
%then
%there exists $X \in \{A, B, C\}$ such that $X<D$ in every  ordering of ${\C}$, or $D<X$ in every  ordering of ${\C}$.
%
\end{theoreme}

%***ATTENTION permutation de b et c dans l'ordre des propositions
%Nous allons maintenant énoncer le théorème permettant de caractériser les configurations fixes en 3 dimensions.
%\begin{theoreme}
%\label{caract}
%Let ${\C}_T$ be a configuration of 3 total orders and let $\mathcal{P}=\{A,B,C,D\}$ be a set of points. Then a), b) and c) are equivalent :
%
%a) ${\C}_T$ is fixed
%
%b) There exists a triplet of points of $\mathcal{P}$ such that all the projections of this triplet are fixed. According to the Proposition~\ref{tripletfixe}, even if it means renumber the points, permute the orders and reverse the orders, we can supposed that the restriction of ${\C}_T$ to this triplet is  
%\begin{displaymath}
%\begin{array}{c c c}
%x_B & x_C & x_A \\
%y_C & y_A & y_B \\
%z_A & z_B & z_C
%\end{array}
%\end{displaymath}In this case there exists $X \in \{A, B, C\}$ such that $X<D$ in all the orders of ${\C}_T$ or $D<X$ in all the orders of ${\C}_T$
%
%c)% EN TRAVAUX
% there exists a developable matrix $M_1$ of the configuration ${\C}_T$ and a developing column $j$ such that the sign of an algebraic development of $\det(M_1)$ by $j$ on the configuration ${\C}_T$ belongs to $\{\plus,\moins\}$.
%\end{theoreme}

%***ATTENTION: ou est le th pour non-fixed ?***

%
%Nous avons programmer ce résultat pour déterminer le nombre de configurations fixes. Au total il n'y a que 5 configurations d'ordres totaux fixes à permutations des axes, inversions d'ordres et renumérotations des points près. Ces configurations sont les suivantes : 

\begin{theoreme}
%\begin{proposition}
\label{caract-non-fixed}
Let $\C$ be a linear ordering configuration on $(\E,\B)$ with $n=4$. Then $\C$ is non-fixed if and only if conditions of Lemma \ref{non-fixe} are satisfied, that is:
there exist 
$e\in\E$ and $b\in \B$  such that
%a projection ${\C}'$ of ${\C}$ w.r.t. $b\in\B$ restricted to $\E\setminus \{e\}$
the configuration $\C'$ induced by $\C$ on $(\E\setminus \{e\},\B\setminus\{b\})$
%, for some $b\in\B$, such that ${\C}'$ 
is non-fixed and
 $e$ is extreme in the ordering $<_b$ of ${\C}$.
\end{theoreme}
%\end{proposition}

%*** ATTENTION BIEN VERIFIER QUE TOUTE CONFIG NON-FIXE PEUT etre detetee ainsi, i.e. celles ded la preuve les engendrent elles bien toutes, l'equivalence ne pose t'elle pass probleme ? ou bien faut-il changer l'enonce ci-dessus en disant : "`is equivalent to a configuratin in whihc conditions f the lemma are satisfied"' ?***

%*** cette prop merite t'elle d'detre un theoreme ? ***

Given the integer $n$ and the sets $\E$ and $\B$ as previously, there are $(n!)^{n-1}$ linear ordering configurations on $(\E,\B)$.
%V23 ajout du n=6 + ajout precision en-dessous + allonger la phrase en dessous
We computed the number of classes of linear ordering configurations up to equivalence from $n=2$ to $n=6$, yielding the sequence:  1, 2, 21, 5097, 71965235.
Implementation details for this computation (required for the case $n=6$) are given in \cite{these}.
%We have no general formula, it would be interesting to find one.
%
This integer sequence  has been added to  
The On-Line Encyclopedia of Integer Sequences 
%Foundation database
%OEIS Foundation database
\cite{oeisf}.
%: \url{http://oeisf.org/}.
%
We have no general formula: we leave as an open question to find one.
%\medskip

\begin{question}
Find a general formula to count, for any $n$, the number of  classes of linear ordering configurations up to equivalence.
\end{question}

We computed the  result provided by Theorem \ref{caract} to list the fixed linear ordering configurations when $n=4$. We found that there are exactly 4 fixed configurations among the 21 linear ordering configurations up to equivalence:

%%%%% CI-DESSOUS AVEC FOOTNOTE
%We computed the  result provided by Theorem \ref{caract} to list the fixed linear ordering configurations when $n=4$. We get that there are exactly 4 fixed configurations among the 21 linear ordering configurations up to equivalence%
%\footnote{%
%%We mention that there are 5097  linear ordering configurations up to equivalence when $n=5$. 
%%
%We computed the number of linear ordering configurations up to equivalence from $n=2$ to $n=5$, yielding the sequence:  1, 2, 21, 5097.
%We have no general formula.
%This integer sequence  has been added to the  OEIS Foundation database: \url{http://oeisf.org/}.% 
%}:

%***ATTENTION: on pourrait integrer d'ordering reversion dans la notion d'equivalence ? mais peut-etre que des fois on ne doit utiliser que les permutations de $\E$ et $\B$...? a voir!***

%DESSOUS configs equivalentes mais pas en ABCD
%\begin{displaymath}
%\begin{array}{c}
%D <_x B <_x C <_x A\\
%D <_y C <_y A <_y B\\
%D <_z A <_z B <_z C
%\end{array}
%\hspace{2.5cm}
%\begin{array}{c}
%D <_x B <_x C <_x A\\
%D <_y C <_y A <_y B\\
%A <_z D <_z B <_z C
%\end{array}
%%\hspace{2.5cm}
%\end{displaymath}
%\bigskip
%\begin{displaymath}
%%%%DESSOUS EFFACEE EQUIVALENTE A DEUXIEME
%%\begin{array}{c}
%%D <_x B <_x C <_x A\\
%%C <_y D <_y A <_y B\\
%%A <_z D <_z B <_z C
%%\end{array}
%\begin{array}{c}
%D <_x B <_x C <_x A\\
%D <_y C <_y A <_y B\\
%A <_z B <_z D <_z C
%\end{array}
%\hspace{2.5cm}
%\begin{array}{c}
%D <_x B <_x C <_x A\\
%C <_y A <_y D <_y B\\
%A <_z D <_z B <_z C
%\end{array}
%\end{displaymath}

\begin{displaymath}
\begin{array}{c}
B <_x C <_x A <_x D\\
C <_y A <_y B <_y D\\
A <_z B <_z C <_z D
\end{array}
\hspace{1cm}
\begin{array}{c}
B <_x C <_x D <_x A\\
C <_y A <_y B <_y D\\
A <_z B <_z C <_z D
\end{array}
%}
\end{displaymath}
%\bigskip
%\bigskip

\begin{displaymath}
\begin{array}{c}
B <_x D <_x C <_x A\\
C <_y A <_y B <_y D\\
A <_z B <_z C <_z D
\end{array}
\hspace{1cm}
\begin{array}{c}
B <_x C <_x D <_x A\\
C <_y D <_y A <_y B\\
A <_z B <_z C <_z D
\end{array}\ 
\end{displaymath}

The interest of the results of this section is that it provides a combinatorial characterization as well as an %(easily computable) 
algorithm capable of deciding if a configuration is fixed or not.
%
%Also, we point out that our result statements deal with being fixed or not, but not with the exact value $\plus$ or $\moins$ of the considered fixed configuration.
We also need to point out that our result statements concern the fixity (or lack thereof) of the considered configuration, but not its exact $\plus$ or $\moins$ value.  This sign can be easily derived from the construction stating the fixity. This sign can also be obtained by choosing any set of points $\P$ satisfying the configuration and evaluating the sign of the real number $det(M(\P))$.
Finally, from the list of fixed linear ordering  configurations given above, one may compute the list of all fixed (partial) ordering  configurations using Proposition \ref{prop:extension}. We do not give this list here.

%%%%%%%%%%%%%%%%%%%%%%%%%%%%%%%%%%%%%%%%%%%%%%%%%%%%%%%%%ù
%%%%%%%%%%%%%%%%%%%%%%%%%%%%%%%%%%%%%%%%%%%%%%%%%%%%%%%%%%%%

%\section{Example continued}
\subsection{Example continued}
\label{sec:ex-suite}

%\begin{figure}
%\includegraphics[height=5.5cm]{./figures/10ptsCrane.png}
%	\caption{10 points sur un crâne}
%	\label{10ptsCrane}
%\end{figure}

%Considérons 10 points de $\mathbb{R}^3$ pris sur un crâne comme le montre la Figure~\ref{10ptsCrane}. On peut alors munir $\mathbb{R}^3$ d'un repère $(O,\vec{x},\vec{y},\vec{z})$ où l'axe $\vec{x}$ va de la droite du crâne vers la gauche du crâne (donc de la gauche vers la droite lorsqu'on regarde les points vus de face) ; l'axe $\vec{y}$ va du bas vers le haut et l'axe $\vec{z}$ va de l'avant du crâne vers l'arrière. Les Figures~\ref{10ptsFace} et \ref{10ptsProfil} montrent respectivement les points vus en face et de profil. En regardant la position relative des points pour un grand nombre de crânes, il apparaît que certaines relations sont respectées pour tous les crânes. Par exemple le point 9 (l'oreille interne droite) sera toujours plus à droite, plus en haut et plus à l'arrière du point 5 (partie droite du menton).

%\begin{figure}[!ht]

%\begin{figure}[h]
%
%\begin{minipage}[c]{0.60\linewidth}
%	\centering
%	%	\includegraphics[height=5cm]{./figures/10ptsFace.pdf}
%		\includegraphics[height=5cm]{./figures/10ptsFaceGrille.pdf}
%	\caption{Points vus de face}
%	\label{10ptsFace}
%\end{minipage}
%
%\begin{minipage}[c]{0.60\linewidth}
%	\centering
%		%\includegraphics[height=5cm]{./figures/10ptsProfil.pdf}
%		\includegraphics[height=5cm]{./figures/10ptsProfilGrille.pdf}
%	\caption{Points vus de profil}
%	\label{10ptsProfil}
%\end{minipage}
%
%\end{figure}

Let us apply the previous results to several configurations in the 3D model shown in Section \ref{sec:example}.
We recall that ordering configurations are represented in Figures \ref{10ptsFace} and \ref{10ptsProfil}, with $\E$  being any set of four points, and $\B$ corresponding to the three axis $\{x,y,z\}$.
\bigskip

\noindent
{\it Example 1. Fixed linear ordering configurations providing a fixed partial ordering configuration:}
{\it the configuration on $\E=\{2,5,8,9\}$  is fixed.}

   This configuration is given by the orderings:
\begin{displaymath}
\begin{array}{c}
    9 <_x 5 <_x 2 <_x 8\\

   5 <_y 8 <_y 9 <_y 2\\

    2 <_z 8 <_z 9\ \ \hbox{ and }\ \ 5 <_z 8 <_z 9\\
   \end{array}
\end{displaymath}

Its two linear extensions, respectively $\C_1$ and $\C_2$, are the following:
%ces deux extensions linéaires (notées respectivement $\mathcal{L}_1$ et $\mathcal{L}_2$) sont :

\begin{minipage}[c]{0.45\linewidth}
\begin{displaymath}
\begin{array}{c}
    9 <_x 5 <_x 2 <_x 8\\

    5 <_y 8 <_y 9 <_y 2\\
    2 <_z 5 <_z 8 <_z 9\\
   \end{array}
\end{displaymath}
\end{minipage}
\begin{minipage}[c]{0.45\linewidth}
\begin{displaymath}
\begin{array}{c}
    9 <_x 5 <_x 2 <_x 8\\

    5 <_y 8 <_y 9 <_y 2\\

    5 <_z 2 <_z 8 <_z 9\\
   \end{array}
\end{displaymath}
\end{minipage}

Let us write these orderings another way:
%Those configurations are respectively equal, up to permutation of $\B$, to
%We prove that they are fixed by Theorem \ref{caract}.

\begin{minipage}[c]{0.45\linewidth}
\begin{displaymath}
\begin{array}{c}
    2 <_z 5 <_z 8 <_z 9\\
   
    5 <_y 8 <_y 9 <_y 2\\
  
    9 <_x 5 <_x 2 <_x 8\\
   \end{array}
\end{displaymath}
%and
\end{minipage}
\begin{minipage}[c]{0.45\linewidth}
\begin{displaymath}
\begin{array}{c}
    5 <_z 2 <_z 8 <_z 9\\

    5 <_y 8 <_y 9 <_y 2\\

    9 <_x 5 <_x 2 <_x 8\\
   \end{array}
\end{displaymath}
\end{minipage}

In this way, we see that, up to a permutation of $\B$ (that is for $\{i,j,k\}=\{x,y,z\}$)
and if we choose
$ A=9$, $B=2$, $C=8$ and $D=5$, then
the orderings in those configurations both satisfy:
\begin{displaymath}
\begin{array}{c}
    B<_i C <_iA\\

    C <_jA<_j B\\

    A <_kB <_kC\\
   \end{array}
\end{displaymath}
as required by Theorem \ref{caract}.
Moreover, for each of these orderings,
$D$ is smaller than $C$ (i.e. $5<_x8$, $5<_y8$, $5<_z8$).
Therefore, according to  Theorem \ref{caract}, those two configurations are fixed.
 It follows that  $\C$ is fixed by Proposition \ref{prop:extension}.

%Ainsi toutes les extensions linéaires de  $\mathcal{C}$ sont fixes donc d'après la proposition 1 $\mathcal{C}$ est fixe.

%***ATTENTION le up to n'est pacs clair ci-dessus, car on change l'ordre des lignes mais sans changer les noms... or il sera confus de changer les noms des ordres... trouver une solution bien lisible***
\bigskip

\noindent
{\it Example 2. A non-fixed ordering configuration implied by a non-fixed linear ordering configuration:}
{\it the configuration on $\E=\{1,3,7,10\}$ is non-fixed.}

It is given by the orderings:
    \begin{displaymath}
\begin{array}{c}
    7 <_x 3 <_x 10 \ \ \hbox{ and }\ \ 7 <_x 1 <_x 10\\
 
    7 <_y 3 <_y 1 \ \ \hbox{ and } \ \ 7 <_y 10 <_y 1\\
  
    1 <_z 7 <_z 10 \ \ \hbox{ and } \ \ 3 <_z 7 <_z 10\\
   \end{array}
\end{displaymath}

 One of its linear extensions is $\mathcal{C'}$:

    \begin{displaymath}
\begin{array}{c}  
    7 <_x 3 <_x 1 <_x 10\\
  
    7 <_y 10 <_y 3 <_y 1\\
  
    3 <_z 1 <_z 7 <_z 10\\
    
   \end{array}
\end{displaymath}

%The projection of $\mathcal{C'}$ w.r.t. $z$ restricted to  $\{7,3,1\}$ is
 The configuration induced by $\C'$ on  $(\{7,3,1\},\{x,y\})$ is
 
      \begin{displaymath}
\begin{array}{c} 
    7 <_x 3 <_x 1\\
  
   7 <_y 3 <_y 1\\
   \end{array}
\end{displaymath}
which is non-fixed by Theorem \ref{non-fixe2D}.
Since 10 is extreme in the ordering $<_z$ of configuration $\C'$,
 $\C'$ is non-fixed by Lemma \ref{non-fixe}, and so is $\C$ by Proposition \ref{prop:extension}.

\bigskip

{\it Example 3.}
We leave as an exercise to check, using Proposition \ref{prop:extension} and Theorem \ref{caract}, that the configuration on $\{1,5,9,10\}$ from Section \ref{sec:example} is fixed. There are four linear extensions to consider, up to symmetries.
\bigskip

Let us conclude by considering  the entire Section~\ref{sec:example} example with ten points in $\mathbb{R}^3$.
Since there is one configuration for each set of 4 points, there are $\binom{4}{10}=210$ configurations to study. We wrote a program to test the fixity of these configurations. For each configuration $\C$ the program lists all the linear extensions of $\C$ and computes if each linear extension is fixed or not, based on the results given in Section~\ref{sec:dim-3}. Then, Proposition~\ref{prop:extension} allows us to conclude. Finally, we find 20 fixed configurations among the 210 configurations. This highlights the significant role of these 20 configurations for 3D skull shape generic characterization, and their non-significant role for the sake of 3D skull shape comparison.

\section{Proofs of Theorem \ref{caract} and Theorem \ref{caract-non-fixed}}
\label{sec:proof}

To prove Theorem~\ref{caract}, we study separately: first, the configurations on $(\E,\B)$ with $n=4$ for which there exists a triplet of points $\E'\subseteq \E$ such that all the 
configurations induced by $\C$ on $\E'$ 
%projections of $\C$ restricted to $\E'$ 
are fixed (characterized by Proposition \ref{tripletfixe}); and second, the other configurations.
The fixed configurations in the first case will be identified.
Then every other configuration in the first case and every configuration in the second case will be proved to be non-fixed, always using Lemma \ref{non-fixe}.
Hence, Theorem \ref{caract-non-fixed} will be proved in the meantime.

\subsection{If all the configurations induced on some triplet are fixed}
%\paragraph{The configurations which there exists a triplet of points such that all the projections are fixed}

Recall that Proposition \ref{tripletfixe} characterizes
configurations for which all the configurations induced on some given triplet are fixed.

\begin{proposition}
\label{prop:proj-fixes}
%Let $\C$ be a configuration on $(\E,\B)$ with $n=4$ and let $\mathcal{P}=\{A,B,C,D\}$ be a set of points such that $\mathcal{P}$ satisfying $\C$. %Let us presume that there exists a triplet of points (we can number this triplet $\{A, B, C\}$) such that all the projections of this triplet are fixed. According to the Proposition~\ref{tripletfixe}, up to equivalence, we can supposed that the restriction of $\C$ to $\{A, B, C\}$ is 
%We suppose that the restriction of $\C$ to $\{A, B, C\}$ is 
%\begin{displaymath} 
%\begin{array}{c c c}
%x_B & x_C & x_A \\
%y_C & y_A & y_B \\
%z_A & z_B & z_C
%\end{array}
%\end{displaymath}
Let $\C$ be a configuration on $(\E,\B)$ with $n=4$, $\E=\{A,B,C,D\}$ and $B=\{x,y,z\}$ such that:
\begin{displaymath}
\begin{array}{c c c c c}
B &<_x& C &<_x& A \\
C &<_y& A &<_y& B \\
A &<_z& B &<_z& C
\end{array}
\end{displaymath}
Then $\C$ is formally fixed by expansion if and only if there exists $X \in \{A, B, C\}$ such that 
%either $X<D$ for every  ordering $<$ in ${\C}$, or $D<X$ for every  ordering $<$ in ${\C}$.
either $X<_b D$ for every $b\in \B$, or $D<_b X$ for every $b\in \B$.
\end{proposition}

\begin{proof}
%fixe : développement par le triplet $ABC$
$\Leftarrow$) 
%Assume that there exists $X \in \{A, B, C\}$ such that $X <_i D$ for all $i \in \E$ or such that $D <_i X$ for all $i \in \E$. We want to prove that $\C$ is fixed. 
%We consider the following formal expression for $det(M)$
Let us denote $\E'=\{A,B,C\}$ and $\C_b$ the configuration induced by $\C$ on 
$(\E',\B\setminus\{b\})$ for $b\in\B$.
%projection of $\C$ w.r.t. $b\in\B$ restricted to $\E'$.
The expansion of $det(M)$ w.r.t. $(D,X)$ yields:
$$det(M)=x_{D-X}\cdot det(M_{\E',\B\setminus \{x\}})
-y_{D-X}\cdot det(M_{\E',\B\setminus \{y\}})
+z_{D-X}\cdot det(M_{\E',\B\setminus \{z\}}).$$
We have 

\begin{eqnarray*}
\det(M_{\E',\B\setminus \{x\}})&=&\det\left(\begin{array}{c c c}
1&1&1\\
%x_A&x_B&x_C&x_D\\
y_A&y_B&y_C\\
z_A&z_B&z_C
\end{array}\right)\\
 & = &\det\left(\begin{array}{c c}
 y_{B-A}&y_{C-A}\\
 z_{B-A}&z_{C-A}
 \end{array}\right)\\
  & = & (y_B-y_A)\cdot(z_C-z_A)-(z_B-z_A)\cdot(y_C-y_A).
\end{eqnarray*}
whose formal sign w.r.t. $\C_x$ is:
$$(\plus\cdot\plus) - (\plus\cdot\moins)=\plus.$$

Similarly, we have
$$\det(M_{\E',\B\setminus \{z\}})
=\det\left(\begin{array}{c c}
 x_{B-A}&x_{C-A}\\
 y_{B-A}&y_{C-A}
 \end{array}\right)
 =
(x_B-x_A)\cdot(y_C-y_A)-(y_B-y_A)\cdot(x_C-x_A)$$
whose formal sign w.r.t. $\C_z$ is 
$$(\moins\cdot\moins) - (\plus\cdot\moins)=\plus.$$

And we have
$$\det(M_{\E',\B\setminus \{y\}})=det\left(\begin{array}{c c}
 x_{B-A}&x_{C-B}\\
 z_{B-A}&z_{C-B}
 \end{array}\right)=(x_B-x_A)\cdot(z_C-z_B)-(z_B-z_A)\cdot(x_C-x_B)$$ whose formal sign w.r.t. $\C_y$ is $$(\moins \cdot \plus) - (\plus \cdot \plus)=\moins.$$
 
Now, if the formal signs of $x_{D-X}$, $y_{D-X}$ and $z_{D-X}$ are all positive (resp. negative), then the formal sign of the above expression of $det(M)$ w.r.t. $\C$ is $\plus$ (respectively $\moins$), which proves that $\C$ is formally fixed by expansion.

\smallskip

%non-fixes : 5 cas se montrent par le Lemme~\ref{non-fixe}
$\Rightarrow$) 
%We will prove the contrapositive i.e. if there exists no $X \in \{A, B, C\}$ such that $X <_b D$ for all $b \in \B$ or such that $D <_b X$ for all $b \in \B$ then $\C$ is non-fixed. 
We will prove the contrapositive: we assume that there exists no $X \in \{A, B, C\}$ such that $X <_b D$ for all $b \in \B$ or such that $D <_b X$ for all $b \in \B$, and we want to prove that $\C$ is non-fixed. 
%
%EMEv9 for every (pas for all) ... there exist (pas exists)
Equivalently, we assume that, for every $X \in \{A, B, C\}$, there exist two orderings in $\C$ such that $X$ is smaller than $D$ in an ordering and $D$ is smaller than $X$ in the other ordering. 
%Without loss of generality, let us assume that $X<_i D$ and $D<_j X$. 
%
Let us consider two cases.
Observe that we will always use Lemma \ref{non-fixe} to prove that $\C$ is non-fixed.
\smallskip

\noindent{\it Case 1:} 
%{\it there exist two orderings $i,j \in \B$ such that, for every $X \in \{A, B, C\}$, we have $X<_i D$ and $D<_j X$.}
{\it there exist two orderings $<_i$ and $<_j$ in $\C$, for $i,j \in \B$, such that, for every $X \in \{A, B, C\}$, we have
either $X<_i D$ and $D<_j X$, or $X<_j D$ and $D<_i X$.}
%$X$ is smaller than $D$ in one ordering of $\{i,j\}$ and $D$ is smaller than $X$ in the other ordering of $\{i,j\}$. 
With the restrictions of $\C$ to $\{A,B,C\}$  given in the hypothesis of the proposition, it is easy to check that  only 3 configurations satisfy this assumption, up to equivalence (i.e. up to some permutations of $\B$ and $\E$, and up to ordering reversions).
%EMEv9 ai mis up to equivalence, ok ?
%Up to a permutation of $\B$, a relabelling of $\{A, B, C\}$ and some reversions of orderings, there exists 3 configurations.\\

\begin{minipage}[c]{0.5\linewidth}
\begin{displaymath}
\begin{array}{c c c c c c c}
B &<_x& D &<_x& C &<_x& A \\
C &<_y& A &<_y& D &<_y& B \\
& A &<_z& B &<_z& C &
\end{array}
\end{displaymath}
\end{minipage}
\begin{minipage}[c]{0.5\linewidth}
\begin{displaymath}
\begin{array}{c c c c c c c}
D &<_x& B &<_x& C &<_x& A \\
C &<_y& A &<_y& B &<_y& D \\
& A &<_z& B &<_z& C &
\end{array}
\end{displaymath}
\end{minipage}

%\begin{minipage}[c]{0.35\linewidth}
\begin{displaymath}
\begin{array}{c c c c c c c}
D &<_x& B &<_x& C &<_x& A \\
& C &<_y& A &<_y& B & \\
A &<_z& B &<_z& C &<_z& D
\end{array}
\end{displaymath}
%\end{minipage}\\
We denote these configurations $\C_1$, $\C_2$ and $\C_3$. In these configurations, one ordering is a partial ordering: 
%We use a partial ordering because whatever the position of $D$ in this ordering, the proof is the same. 
each of these configurations represents four linear ordering configurations.
%We use a partial ordering because whatever the position of $D$ in this ordering, the proof is the same, hence we group 4 linear ordering configurations together for every $C_i$ with $i\in \{1,2,3\}$.
%EMEv9 attention i deja utilisé
%EMEv9 tu dis que la preuve est la meme, et dessous tu distingues les cas ou A<_z D et A>_z D... mal dit donc

Let us prove that $\C_1$ is non-fixed. 
First, assume that $A <_z D$. 
The configuration induced by $\C_1$ on $(\{B,C,D\},\{x,y\})$ is:
%The projection of $\C_1$ with respect to $z$ and restricted to $\{B,C,D\}$ is:
\begin{displaymath}
\begin{array}{c c c c c}
B &<_x& D &<_x& C \\
C &<_y& D &<_y& B
\end{array}
\end{displaymath}
%EMEv9 ai ajoute "`by theorem"' ci-dessous, toujours justifier en citant les resultats anterieurs meme simples, sauf lorsqu'on emploie toujours le meme argument !
Since this configuration is non-fixed by Theorem \ref{non-fixe2D},
and since $A$ is minimal in the ordering $<_z$, then $\C_1$ is non-fixed by Lemma~\ref{non-fixe}. 
Second, similarly, if $D <_z A$, then 
the configuration induced by $\C_1$ on $(\{A,C,D\},\{y,z\})$ is non-fixed  by Theorem \ref{non-fixe2D}:
%EMEv9 typo c'est y,z qu'il faut considerer
\begin{displaymath}
\begin{array}{c c c c c}
C &<_y& A &<_y& D \\
D &<_z& A &<_z& C
\end{array}
\end{displaymath}
%the projection of $\C_1$ with respect to $y$ and restricted to $\{A,C,D\}$ is
And $B$ is minimal in $<_x$, so $\C_1$ is non-fixed by Lemma~\ref{non-fixe}.

Now we will show that $\C_2$ is non-fixed. The proof is similar. 
Assume that $A <_z D$. As before, $A$ is minimal in the ordering $<_z$. 
The configuration induced by $\C_2$ on $(\{B,C,D\},\{x,y\})$
%The projection of $\C_2$ with respect to $z$ and restricted to $\{B,C,D\}$ 
is non-fixed, so $\C_2$ is non-fixed by Lemma~\ref{non-fixe}.
%EMEv9 pas "`however"' (=cependant, toutefois)
 Assume that $D <_z A$. 
Then, the configuration induced by $\C_2$ on $(\{A,B,D\},\{x,y\})$
 %the projection of $\C_2$ with respect to $z$ and restricted to $\{A,B,D\}$ 
 is non-fixed. Since $C$ is maximal in the ordering $<_z$, we find that $\C_2$ is non-fixed by Lemma~\ref{non-fixe}.

We will use the same method to prove that $\C_3$ is non-fixed. First, assume that $C <_y D$. 
The configuration induced by $\C_3$ on $(\{A,B,D\},\{x,z\})$
%The projection of $\C_3$ with respect to $y$ and restricted to $\{A,B,D\}$ 
is non-fixed. Since $C$ is minimal in the ordering $<_y$, then $\C_3$ is non-fixed by Lemma~\ref{non-fixe}. Secondly, if $D <_y C$ then 
the configuration induced by $\C_3$ on $(\{B,C,D\},\{y,z\})$
%the projection of $\C_3$ with respect to $x$ and restricted to $\{B,C,D\}$ 
is non-fixed. Since $A$ is maximal in the ordering $<_x$, then $\C_3$ is non-fixed by Lemma~\ref{non-fixe}.
\smallskip

%EMEv9 cas 2 etait un peu abrupt
\noindent{\it Case 2:} {\it the assumption of case 1 does not hold}.
It is easy to check that, up to equivalence of configurations, there are two other configurations such that $X$ is smaller than $D$ in an ordering and $D$ is smaller than $X$ in another ordering:

\begin{minipage}[c]{0.45\linewidth}
\begin{displaymath}
\begin{array}{c c c c c c c}
D &<_x& B &<_x& C &<_x& A \\
C &<_y& A &<_y& D &<_y& B \\
A &<_z& B &<_z& D &<_z& C
\end{array}
\end{displaymath}
\end{minipage}
\begin{minipage}[c]{0.45\linewidth}
\begin{displaymath}
\begin{array}{c c c c c c c}
B &<_x& D &<_x& C &<_x& A \\
C &<_y& D &<_y& A &<_y& B \\
A &<_z& D &<_z& B &<_z& C 
\end{array}
\end{displaymath}
\end{minipage}\\
We denote these configurations $\C_4$ and $\C_5$.% Nous allons de nouveau utiliser le Lemme~\ref{non-fixe} pour montrer que ces configurations d'ordres sont non-fixes.

%EMEv9 "`start off"' me semble douteux
%We start off with proving 
We first prove
that $\C_4$ is non-fixed. 
The configuration induced by $\C_4$ on $(\{B,C,D\},\{y,z\})$ 
%The projection of $\C_4$ with respect to $x$ and restricted to $\{B,C,D\}$ 
is non-fixed (again by Theorem \ref{non-fixe2D}), and $A$ is maximal in the ordering $<_x$. So the configuration $\C_4$ is non-fixed by Lemma~\ref{non-fixe}.

Similarly, the configuration induced by $\C_5$ on $(\{B,C,D\},\{x,y\})$ 
%The projection of $\C_5$ with respect to $z$ and restricted to $\{B,C,D\}$ 
is non-fixed. Since $A$ is minimal in the ordering $<_z$, $\C_5$ is non-fixed by Lemma~\ref{non-fixe}.
\end{proof}

\subsection{If for every triplet there is at least one non-fixed induced configuration}
%\paragraph{The configurations such that for all triplet there is at least one non-fixed projection}

%EMEv9 ajout de paragraphe ci-dessous

We call \emph{triplet} a set of three elements. For short, we may denote $ABC$ the triplet $\{A,B,C\}$.
Consider two linear orderings $<_i$ and $<_j$  on a same set containing the triplet $\{A,B,C\}$ and assume, without loss of generality, that $A<_i B <_i C$.
We say that $<_i$ and $<_j$  are \emph{equal}, resp. \emph{reversed}, \emph{on the triplet $\{A,B,C\}$}, if $A<_j B <_j C$, resp.  $C<_j B <_j A$. 
Recall that, in this entire section, according to Theorem \ref{non-fixe2D}, 
$<_i$ and $<_j$  are equal or reversed on $\{A,B,C\}$ if and only if 
the configuration formed by $<_i$ and $<_j$ on $(\{A,B,C\},\{i,j\})$ is non-fixed.
In what follows, we will often consider two such orderings equal or reversed on such a triplet. 
For short, we may also say that the triplet $\{A,B,C\}$ is \emph{conformal} w.r.t. the two orderings $<_i$ and $<_j$.
%Also, for short, we may denote $ABC$ the triplet $\{A,B,C\}$.

%On regarde 2 ordres de ${\C}$ ayant le plus de triplets de points qui sont vus dans le même sens (ou dans le sens inverse) dans ces 2 ordres.
%\begin{lemma}
%\label{couples}
%Let $\C$ be a configuration on $(\E,\B)$ with $n=4$ and $\E=\{A,B,C,D\}$. 
%%EMEv9 attention syntaxe du let : Let $<_i$ and $<_j$ be two orderings of $\C$. PAS Let two orderings $<_i$ and $<_j$ of $\C$. 
%%Let two orderings $<_i$ and $<_j$ of $\C$. 
%Let $<_i$ and $<_j$ be two orderings of $\C$. 
%%EMEv9 supposer= to assume, We can assume PAS we can supposed
%%We can supposed 
%We can assume
%that the ordering $<_i$ is $A <_i B <_i C <_i D$. Let $(ABC,BCD)$, $(ABC,ACD)$ and $(ABD,BCD)$ three pairs of triplets of $\E$. If for one of this pairs, the two triplets are equally ordered or oppositely ordered in the orderings $<_i$ and $<_j$, then we are either $A <_j B <_j C <_j D$, or $D <_j C <_j B <_j A$.
%\end{lemma}

%EMEv9 attention (ABC,ABD) est un couple, dans lequel l'ordre est important, on veut des paires, donc entre accolades
%EMEv9 enonce de ce lemme etait mal foutu, a comparer
\begin{lemma}
\label{couples}
Let $\C$ be a configuration on $(\E,\B)$ with $n=4$, $\E=\{A,B,C,D\}$ and $\B=\{x,y,z\}$. 
%Let $<_i$ and $<_j$ be two orderings of $\C$. 
Assume that 
$A <_x B <_x C <_x D$.
Consider the three pairs of triplets $\bigl\{\{A,B,C\},\{B,C,D\}\bigr\}$, $\bigl\{\{A,B,C\},\{A,C,D\}\bigr\}$ and $\bigl\{\{A,B,D\},\{B,C,D\}\bigr\}$. 
If, for at least one of these pairs, the two orderings $<_x$ and $<_y$ are
equal or reversed on both triplets of the pair, 
then $<_x$ and $<_y$ are equal or reversed on $\{A,B,C, D\}$, that is:
either $A <_y B <_y C <_y D$ or $D <_y C <_y B <_y A$.

%EMEv9 ci-dessous autre formulation possible de la fin du lemme
%
%If $<_x$ and $<_y$ are equal or reversed on $\{A,B,C\}$ and on $\{B,C,D\}$
%then  $<_x$ and $<_y$ are equal or reversed on $\{A,B,C, D\}$, that is:
%either $A <_y B <_y C <_y D$ or $D <_y C <_y B <_y A$.
%
%The same conclusion holds if
%$<_x$ and $<_y$ are equal or reversed on $\{A,B,C\}$ and on $\{A,C,D\}$.
%And the
%same conclusion holds if
%$<_x$ and $<_y$ are equal or reversed on $\{A,B,D\}$ and on $\{B,C,D\}$.
%
%Let $(ABC,BCD)$, $(ABC,ACD)$ and $(ABD,BCD)$ three pairs of triplets of $\E$. If for one of this pairs, the two triplets are equally ordered or oppositely ordered in the orderings $<_i$ and $<_j$, then we are either $A <_j B <_j C <_j D$, or $D <_j C <_j B <_j A$.
\end{lemma}

\begin{proof}
First, we prove that if two orderings $<_i$ and $<_j$ are equal or reversed on two triplets of $\E$, then  they are either equal on the two triplets, or reversed on the two triplets.
%
%EMEv9  "`Let 2 triplets of $\E$."'  ne veut rien dire
%We will prove that if these two triplets are equally ordered or oppositely ordered in the orderings $<_i$ and $<_j$ then there are either equally ordered in $<_i$ and $<_j$, or oppositely ordered in $<_i$ and $<_j$. 
These two triplets have two elements  in common. We denote these triplets $XYZ$ and $XYW$. Up to relabelling, we can assume that $X <_i Y$. 
%Assume that $XYZ$ is equally ordered $<_i$ and $<_j$ and $XYW$ is oppositely ordered in $<_i$ and $<_j$ (reductio ad absurdum).
Assume for a contradiction that $<_i$ and $<_j$ are equal on $XYZ$ and reversed on $XYW$.
 The triplet $XYZ$ shows that $X <_j Y$, whereas $XYW$ shows that $Y <_j X$,
 which is a contradiction.
 %So we can conclude that two triplets of $\E$ are either equally ordered in $<_i$ and $<_j$, or oppositely ordered in $<_i$ and $<_j$.
%EMEv9 paragraphe ci-dessus reformulé

%EMEv9 "`going back"' ne me sonne pas bien
%EMEv9 paragraphe ci-dessous allégé
Let us turn back to the pairs of triplets $\{ABC,BCD\}$, $\{ABC,ACD\}$ and $\{ABD,BCD\}$. 
Let us first consider the pair $\{ABC,BCD\}$. 
Assume that $<_x$ and $<_y$ are equal on $ABC$ and $BCD$. We prove that $<_x$ and $<_y$ are equal on $\E$.
%Assume that there is a pair of triplets such that the two triplets are equally ordered in $<_i$ and $<_j$. 
%We will prove that we have $A <_j B <_j C <_j D$ if this pair is $(ABC,BCD)$. 
%Then $ABC$ is equally ordered in $<_i$ and $<_j$. 
We have $A <_x B <_x C$ by the lemma's hypothesis, and $<_x$ and $<_y$ equal on $ABC$ by assumption, so we have $A <_y B <_y C$.
%We also have $C<_x D$ by hypothesis, hence $C<_y D$
We have $C <_x D$ by hypothesis, 
and $<_x$ and $<_y$ equal on $BCD$ by assumption,
%and $BCD$ is also equally ordered in $<_i$ and $<_j$ and we have $C <_i D$. 
so we have $C <_y D$. 
Hence, the ordering $<_y$ is $A <_y B <_y C <_y D$. 
If we consider the pair $\{ABC,ACD\}$, then we obtain the same result with the same proof.
Consider the pair $\{ABD,BCD\}$, and assume that $<_x$ and $<_y$ are equal on $ABD$ and $BCD$.  Then the ordering $<_y$ restricted to $ABD$ is $A <_y B <_y D$. The orderings $<_x$ and $<_y$ are also equal on $BCD$. So we have $B <_y C <_y D$. This proves that we have $A <_y B <_y C <_y D$.
% if the two triplets are equally ordered in $<_i$ and $<_j$.

Now, assume that there is a pair of triplets such that 
$<_x$ and $<_y$ are reversed on the two triplets of the pair.
%the two triplets are oppositely ordered in $<_i$ and $<_j$. 
Let us denote $<_{opp(y)}$ the reversion of the ordering $<_y$. 
%The two triplets are equally ordered in $<_i$ and $<_{opp(j)}$. 
Then $<_x$ and $<_{opp(y)}$ are equal on the two triplets of the pair.
%EMEv9 from what we have just seen - ne me sonne pas bien anglais ecrit
%From what we have just seen, 
From the result above, we get $A <_{opp(y)} B <_{opp(y)} C <_{opp(y)} D$, 
that is  $D <_y C <_y B <_y A$.
\end{proof}

%EMEv9 j'ai mis x ett y a la place de i et j car on ne perd pas en generalite et on s'est hbaitue a ces lettres pour designer les elemetns de B

\begin{proposition}
\label{proj_non-fixe=>non-fixe}
Let $\C$ be a configuration on $(\E,\B)$ with $n=4$. If 
%EMEv9 JE CROIS QUE C'EST for every PAS for all, for every = pour chaque, for all = pour tous... je n'ai jamais bien compris la nuance, mais on m'a deja corrige ca... a verifier aupres d'un anglophone...
%for all triplets 
for every triplet
$\E' \subseteq \E$ there exists $b \in \B$ such that 
the configuration induced by $\C$ on $(\E',\B\setminus\{b\})$
%the projection of $\C$ with respect to $b$ and restricted to $\E'$ 
is non-fixed, then $\C$ is non-fixed.
\end{proposition}

\begin{proof}
%EMEv9 pas trop sur de ce "`for ease of the proof"' je le laisse dans le doute
%For ease of the proof, we first separate the configurations. 
%EMEv9 j'ai repris le paragrahe ci-dessous, pas assez precis pour etre  lisible sans peine... compares ! 
We consider separate cases.
To this aim, for each ordering $<_b$ in $\C$, we count the number of non-fixed configurations induced by $\C$ on  $(\E',\B\setminus\{b\})$ 
with $\E'=\E\setminus\{e\}$ for some $e\in\E$.
By Theorem \ref{non-fixe2D}, it amounts to counting, 
for each pair of orderings in $\C$ (i.e. for each $\B\setminus\{b\})$), 
the number of triplets $\E'$ in $\{ABC,ABD,ACD,BCD\}$ which are conformal w.r.t. this pair of orderings. 
Without loss of generality (up to a permutation of $\B$),
we can assume that the pair or orderings $\{<_x,<_y\}$ maximizes this number of triplets.
%Let us denote $(<_i,<_y)$ a pair of orderings of $\C$ maximising this number of triplets, and $<_z$ the third ordering of $\C$. 
%EMEv9 without loss of generality = terme classique bien pratique
Also, without loss of generality (up to a permutation of $\E$),
%For the sake of simplicity,  
%up to a permutation of $\E$,
%we permute the elements of $\E$ in order to have 
we can assume that $A <_x B <_x C <_x D$.
%without loss of generality. 
%We will determine if $\C$ is fixed or non-fixed by considering the different number of triplets conformal w.r.t. $<_i$ and $<_y$.
The different cases correspond to the number of triplets of $\E$ which are conformal w.r.t. $<_x$ and $<_y$.
%EMEv9 
%A-FAIRE pour bien faire il faudrait changer  les i,j,k en x,y,z, la j'ai la flemme
\medskip

%\subparagraph*{4 triplets}
\noindent{\it Case 1: four triplets of $\E$ are conformal w.r.t. $<_x$ and $<_y$.}

All the triplets of $\E$ are conformal w.r.t. $<_x$ and $<_y$. Therefore, we have $A <_y B <_y C <_y D$ or $D <_y C <_y B <_y A$ by Lemma~\ref{couples}. 
%EMEv9 ci dessous precision inutile, c'est le but de la proposition !
%We will prove that $\C$ is non-fixed. 
Let $X \in \E$ be an extreme element in the ordering $<_z$. The triplet $\E'=\{A,B,C,D\}\setminus \{X\}$ is conformal w.r.t. $<_x$ and $<_y$. That is, using Theorem \ref{non-fixe2D}: the configuration induced by $\C$ on $(\E',\{x,y\})$ is non-fixed. Hence, $\C$ is non-fixed by Lemma~\ref{non-fixe}.
\medskip

%\subparagraph*{3 triplets}
\noindent{\it Case 2: three triplets of $\E$ are conformal w.r.t. $<_x$ and $<_y$.}

We prove that this case is not possible.
%EMEv9 intuile de repeter les hypotheses 
%it is not possible to have exactly 3 triplets equally ordered or oppositely ordered in $<_i$ and $<_y$. 
%EMEv9 assume for a contradiction PAS by
%Assume for a contradiction that 3 triplets are equally ordered or oppositely ordered in $<_i$ and $<_y$. 
%EMEv9 ci-dessus, encore repetition... faire moins de phrases "d'annonce"' qui se repetent surtout dans une preuve courte comme ici, mais par contre des phrases plus precises et completes
Otherwise,
%EMEv9 dessous mal dit
%there is one of the pairs $(ABC,BCD)$, $(ABC,ACD)$, $(ABD,BCD)$ among this 3 triplets. 
the set of three triplets contains one of the pairs $\{ABC,BCD\}$, $\{ABC,ACD\}$, $\{ABD,BCD\}$.
Then we have $A <_y B <_y C <_y D$ or $D <_y C <_y B <_y A$ by Lemma~\ref{couples}. Therefore, the four triplets are conformal w.r.t. $<_x$ and $<_y$.
\medskip

%\subparagraph*{2 triplets}
\noindent{\it Case 3: two triplets of $\E$ are conformal w.r.t. $<_x$ and $<_y$.}

%EMEv9
%A-FAIRE : l'ordre des paragraphes ci-dessous purrait peut etr etre ameliore, l'enchaineent n'est pas facile a suivre
%A-FAIRE : mettre indices aux T_1, O_1 etc...

This is the most critical case.
%We will to prove that $\C$ is non-fixed. 
First, we build a bipartite graph $G$ between the four triplets of $\E$ and the three orderings of $\C$. 
%EMEv9 pas  tres net, dessus tu dis between trilets and verticecs, dessous tu dis corresponding to, les sommet SONT les triplets et ordres
%Let us denote $T_1\ T_2\ T_3\ T_4$ the vertices corresponding to the triplets and $O_1\ O_2\ O_3$ the vertices corresponding to the orderings of $\C$. 
%EMEv9 de plus on ne comprend pas pourqoi ils changent de noms... tentative d'amelioration ci-dessous, peut etre pas encroe au top
Let us denote $T_1,\ T_2,\ T_3,\ T_4$ those triplets, and   $O_1,\ O_2,\ O_3$ those orderings (we will study later in our case analysis which label corresponds to which triplet/ordering).
%EMEv9 ai change ci-dessous
Thus, the edges are of the form $(T,O)$ with $T$ in $\{T_1,\cdots,T_4\}$ and $O$ in $\{O_1,O_2,O_3\}$. 
The graph $G$ is defined the following way:
if a triplet $T$ is conformal w.r.t. the two orderings $O_a$ and $O_b$, for $a,b\in\{1,2,3\}$,  then
both $(T,O_a)$ and $(T,O_b)$ are edges of $G$;
and every edge $(T,O_a)$ in $G$ means that there exists $O_b$ such that $T$ is conformal w.r.t.  $O_a$ and $O_b$.
%We add edges by pair 
%$\left((T,O_a)\ , \ (T,O_b)\right)$ 
%%$(\ (T,O_a)$ and $(T,O_b)\ )$ 
%which are as extremity the same triplet if this triplet is equally ordered or oppositely ordered in the two orderings $O_a$ and $O_b$ of $\C$. 
%

Let us state several useful claims. %observations. %claims.

\begin{enumerate}

\item 
\label{claim-iff}
\emph{We have: $(T,O_a)$ and $(T,O_b)$ are edges of $G$ if and only if
$T$ is conformal w.r.t. the two orderings $O_a$ and $O_b$.}
Indeed, if $(T,O_a)$ and $(T,O_b)$ are edges, but $T$ is not conformal w.r.t. the two orderings $O_a$ and $O_b$, then there exists an edge $(T,O_c)$ in $G$, with
$T$  conformal w.r.t. the two orderings $O_a$ and $O_c$, and $T$ conformal w.r.t. $O_b$ and $O_c$. This directly implies that the orderings $O_a$ and $O_b$ are equal or reversed on $T$, that is $T$ is conformal w.r.t. $O_a$ and $O_b$.

\item
\label{claim-ij}
%Since the orderings $<_i$ and $<_y$ are equal or reversed on exactly two triplets, 
\emph{We have: the two vertices $<_x$ and $<_y$ of $G$ share exactly two neighbors.} Indeed, exactly two triplets are conformal w.r.t. $<_x$ and $<_y$ by assumption.
%Since exactly two triplets are conformal w.r.t. $<_x$ and $<_y$ by assumption,
%%EMEv9 dessous incomrehensible
%%the vertices corresponding to these orderings are exactly 2 vertices in their common neighbourhood in the graph.
%the two vertices $<_x$ and $<_y$ of $G$ have exactly two neighbors in common.

\item
\label{claim-pairs}
%EMEv9 eviter d'ecrire des chiffres two orderings, PAS 2 orderings, mieux vaut en toutes lettres, les hciffres ce sont les eniters mathematiques
\emph{We have: if $<_x$ and another vertex in $\{O_1,O_2,O_3\}$ are 
%EMEv9 adjacent PAS linked
adjacent in $G$ to the same two vertices in $\{T_1,\cdots,T_4\}$, 
then these two vertices in $\{T_1,\cdots,T_4\}$ form one of these pairs: $\{ABC,ABD\}$, or $\{ABD,ACD\}$, or $\{ACD,BCD\}$.}
Indeed, if two triplets of $\E$ are conformal w.r.t. the same two orderings $<_x$ and $<_a$ in $\C$, for $a\in\{y,z\}$,
%If in a pair of orderings there are 2 triplets equally ordered or oppositely ordered, 
then these two triplets are not $\{ABC,BCD\}$, nor $\{ABC,ACD\}$, nor $\{ABD,BCD\}$: otherwise, by Lemma \ref{couples}, the four triplets of $\E$ would be
conformal w.r.t. the two orderings, which would contradict the maximality property of $<_x$ and $<_y$ and the assumption of Case 3. 
%So, if $<_x$ and another vertex in $\{O_1,O_2,O_3\}$ are 
%%EMEv9 adjacent PAS linked
%adjacent in $G$ to the same two vertices in $\{T_1,\cdots,T_4\}$, 
%then these two vertices in $\{T_1,\cdots,T_4\}$ form one of these pairs: $\{ABC,ABD\}$, or $\{ABD,ACD\}$, or $\{ACD,BCD\}$.

%Si une paire d'ordres à 2 triplets qui sont ordonnés de manière égale ou opposée, alors ces 2 triplets sont différents des couples du Lemme~\ref{couples} (sinon les 2 ordres seraient égaux où opposés et les 4 triplets seraient ordonnés de manière égale ou opposée). Ainsi lorsque dans le graphe 2 ordres sont reliés au 2 mêmes triplets, ces triplets appartiennent à l'un des couples suivants : $\{ABC,ABD\}$ $\{ABD,ACD\}$ $\{ACD,BCD\}$.

\item
\label{claim-connected}
\emph{We have: in the graph $G$, each vertex in $\{T_1,\cdots,T_4\}$ is not isolated, and is adjacent to at least two vertices in $\{O_1,O_2,O_3\}$.}
Indeed, by hypothesis of the proposition,
%EMEv9 
%In all the configurations, 
for every triplet $T$ of $\E$, there exists $b\in\B$ such that
the configuration induced by $\C$ on $(T,\B\setminus\{b\})$ is non fixed,
that is, as in Theorem \ref{non-fixe2D}, such that $T$ is conformal w.r.t. the two orderings in $\C$ different from $<_b$.
%for all triplet there is at least one non-fixed projection. 
%In other words: each triplet is conformal in at least 2 orderings. 
%So, in the graph $G$, each vertex in $\{T_1,\cdots,T_4\}$ is not isolated, and is adjacent to at least two vertices in $\{O_1,O_2,O_3\}$. 
%

\item
\label{claim-8}
%EMEv9 have degree PAS are degree
%By claim \ref{claim-connected} above, we have that there are at least eight edges in $G$.
\emph{We have: there are at least eight edges in $G$.} This is obtained directly from Claim \ref{claim-connected} above.

\item
\label{claim-degree}
\emph{We have: all the vertices in $\{O_1,O_2,O_3\}$ have degree at least two.}
Indeed, assume that a vertex in $\{O_1,O_2,O_3\}$ has degree zero or one, then, among the two other vertices in $\{O_1,O_2,O_3\}$, one has degree four and the other has a degree at least equal to three. Hence, those two orderings are adjacent to three common triplets. 
We get a contradiction with  the definition of $<_x$ and $<_y$, because this pair has been assumed to maximize the number of triplets conformal w.r.t. it (and we assumed that this number was equal to two).
%So all the vertices in $\{O_1,O_2,O_3\}$ have degree at least two. 

%%EMEv9 have degree PAS are degree
%Finally, in $G$, there are at least eight edges, and all the vertices in $\{O_1,O_2,O_3\}$ have degree at least two. 
%%EMEv9 POURQUOI $\{O_1,O_2,O_3\}$ have degree at least two ?
%Hence there are either at least 1 vertex in $\{O_1,O_2,O_3\}$ of degree 4, or at least 2 vertices in $\{O_1,O_2,O_3\}$ of degree 3. The Figures~\ref{biparti4} and \ref{biparti33} show respectively these 2 bipartite graphs. 
\end{enumerate}

%EMEv9 je reprend la fin de cette preuve, qui a ete un peu baclee ! je ne met plus de commentaires, tu compareras soigneusement. En particulier les figures moontre des cas minimaux, et les autres aretes potentielles doivent etre prises en compte, ce qui n'est pas fait clairement (par exemple tu ne peux affirmer qu'il y a un isomoprhisme sans prendre en compte ces aretes)

By Claim \ref{claim-8} and Claim \ref{claim-degree}, we have that: either there is at least one vertex of $G$ in $\{O_1,O_2,O_3\}$ with degree 4, or there are at least two vertices of $G$ in $\{O_1,O_2,O_3\}$ with degree 3. 
Figures~\ref{biparti4} and \ref{biparti33} show respectively the two  bipartite graphs with a minimal number of edges and satisfying one of these two properties, up to permutations of $\{T_1,...,T_4\}$ and $\{O_1,O_2,O_3\}$. 
So we now assume, without loss of generality, that removing edges from $G$ leads to one the two graphs depicted in Figures~\ref{biparti4} and \ref{biparti33}.
%And let us consider separately the two cases depending on these two properties.

%Or nous ne nous intéressons qu'aux configurations d'ordres totaux pour lesquelles aucun triplet n'a toutes ses projections fixes. Cela signifie que pour chaque triplet, il existe au moins une projection non-fixe. Autrement dit, chaque triplet est ordonné de manière égale ou opposée dans au moins 2 ordres. Au niveau du graphe biparti cela signifie que chaque triplet est relié à au moins 2 ordres. Le graphe a alors au moins 8 arêtes. On a soit un sommet correspondant à un ordre qui est de degré 4 dans le graphe, soit 2 sommet correspondant à des ordres qui sont de degré 3 dans le graphe. On peut voir dans les Figures~\ref{biparti4} et \ref{biparti33} une représentation de ces 2 graphes bipartis.

%\bigskip

\begin{figure}[!ht]
\hfil
\begin{minipage}[t]{0.35\linewidth}
\centering \includegraphics[height=4.5cm]{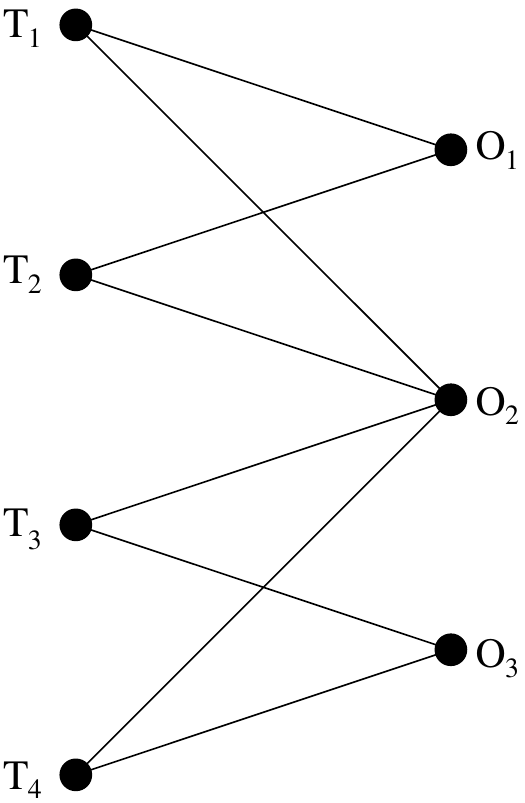}
\caption{$G$ minimal with one vertex having degree 4}\label{biparti4}
\end{minipage}\hfil
\begin{minipage}[t]{0.35\linewidth}
\centering \includegraphics[height=4.5cm]{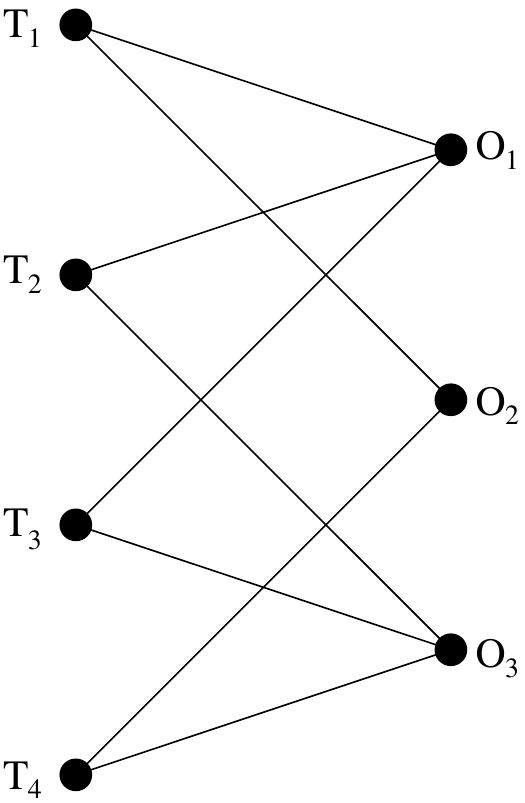}\\
\caption{$G$ minimal with two vertices having degree 3}\label{biparti33}
\end{minipage}
\hfil
\end{figure}

%In what follows, we say that two vertices $a$ and $b$ of $G$ are \emph{equivalent} if there exists an isomorphism $f$ of $G$ such that $f(a)=b$.
%
%Pour alléger la démonstration, nous dirons que 2 sommets $a$ et $b$ d'un graphe $G$ sont \emph{équivalents} s'il existe un isomorphisme $f$ tel que $f(a)=b$.
%Let us consider the two cases depending on the above two properties,

Let us consider the two cases separately, according to these two properties.

\begin{itemize}

\item $G$ has at least one vertex with degree 4 among $\{O_1,O_2,O_3\}$. 
See Figure \ref{biparti4}.
\smallskip

\begin{figure}[!ht]
\hfil
\begin{minipage}[t]{0.35\linewidth}
\centering \includegraphics[height=4.5cm]{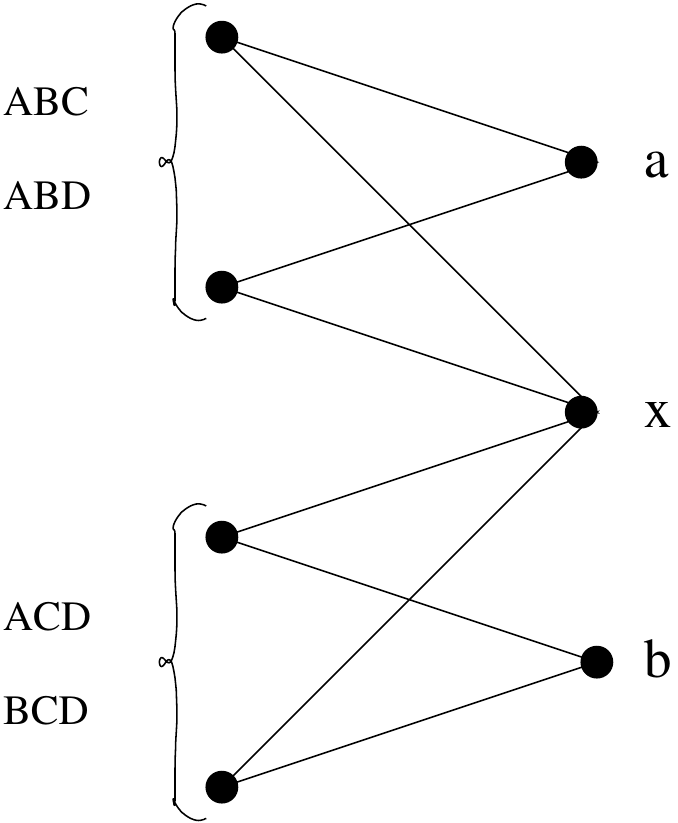}
\caption{$O_2$ is $<_x$}\label{biparti4-ji}
\end{minipage}
\hfil
\begin{minipage}[t]{0.35\linewidth}
\centering \includegraphics[height=4.5cm]{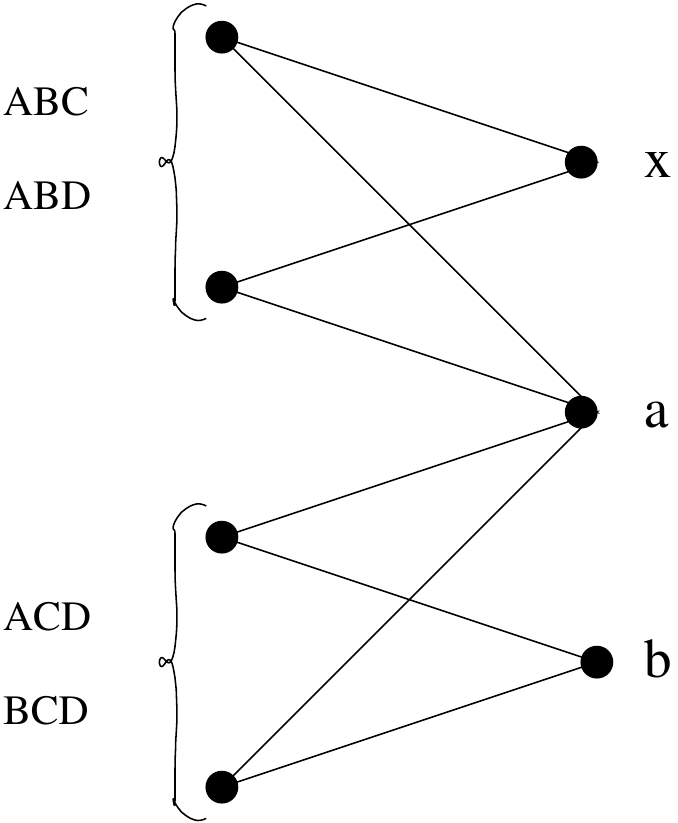}
\caption{$O_1$ is $<_x$}\label{biparti4-ij}
\end{minipage}\hfil
\end{figure}

%\bigskip

%*** Il faudra CHNAGER i,y,z en x,y,z partout et sur figure
%(puisque i,y sont speciifes des le debut, autant les appeler x,y !
%--- JE LE FERAIS A LA TOUTE FIN, sinn je vais m'emboruiller

First, assume that the ordering $<_x$ is the vertex $O_2$ with degree four. 
%
%*** dessous Mise a jour recuperee du passage faux dans ce cas
%
In this case, the two triplets $T_1$ and $T_2$ are adjacent to the same two orderings $O_1$ and $<_x$. By Claim \ref{claim-pairs} above, this implies that
$T_1$ and $T_2$ form one of these pairs: $\{ABC,ABD\}$, $\{ABD,ACD\}$, $\{ACD,BCD\}$.
The same result holds for the two triplets $T_3$ and $T_4$ adjacent to $O_3$ and $<_x$.
%In the Figure~\ref{biparti4}, we show that the vertices $T_1$ and $T_2$ are equivalent, and the vertices $T_3$ and $T_4$ too. It is the same for the vertices $O_1$ and $O_3$. So the pairs of vertices $(T_1,T_2)$ and $(T_3,T_4)$. The vertices $T_1$ and $T_2$ share 2 neighbors in common so the pair $(T_1,T_2)$ corresponds to one of these pairs: $(ABC,ABD)$, $(ABD,ACD)$, $(ACD,BCD)$. 
%It is the same for $(T_3,T_4)$. 
%Since these vertices correspond to distinct triplets, $(T_1,T_2)$ and $(T_3,T_4)$ are not $(ACD,BCD)$. 
Since these four triplets are distinct, $\{T_1,T_2\}$ and $\{T_3,T_4\}$ are not equal to $\{ABD,ACD\}$.
%The pairs of vertices $(T_1,T_2)$ and $(T_3,T_4)$ are equivalent so, 
Without loss of generality, we can assume that $\{T_1,T_2\}=\{ABC,ABD\}$ and $\{T_3,T_4\}=$ $\{ACD,$ $BCD\}$. 

%****** fin passage reporté

%Since $(T_1,T_2)$ and $(T_3,T_4)$ are equivalent, we can assume without loss of generality that $O_1$ correspond to $<_y$. 

%\hspace{3cm}\begin{minipage}[c]{0.75\linewidth}
%{\pti
%***DEBUT ERREUR
%
%Since $(T_1,T_2)$ and $(T_3,T_4)$ are equivalent, we can assume without loss of generality that $O_1$ correspond to $<_y$. 
%
%
%*** ATTENTION, NON SYMETRIQUE, CAR $<_y$ SPECIFIE pour former couple aximal avec $<_x$ DONC IL MANQUAIT LE CAS OU  $<_y=O_3$ !
%on s'en sort facilement en ne precisant pas que c'est $y$
%
%****** FIN ERREUR
%}
%\end{minipage}

Let us denote $a,b\in\{y,z\}$ such that $<_a$ is $O_1$,
and $<_b$ is $O_3$.
This case is illustrated in Figure~\ref{biparti4-ji}. 
%
%*** changer figure y -> a
%
%
In this case, $ABC$ is adjacent in $G$ to $<_x$ and $<_a$. 
Hence, based on Claim \ref{claim-iff}, 
the configuration induced by $\C$ on $(ABC,\{x,a\})$ is non-fixed.
%the projection of $\C$ w.r.t. $z$ restricted to $ABC$ is non-fixed. 
On the other hand, 
 $ACD$ is adjacent to $<_x$ and $<_b$ in $G$, meaning that 
 $<_x$ and $<_b$ are equal or reversed on $ACD$. Since $D$ is extreme in $<_x$, then $D$ is extreme in $<_b$ restricted to $ACD$.
 Similarly, $BCD$ is adjacent to $<_x$ and $<_b$ in $G$ implies that
 $D$ is extreme in $<_b$ restricted to $BCD$. Hence $D$ is extreme in $<_b$. So we find that $\C$ is non-fixed by Lemma~\ref{non-fixe}.

Second, assume that the ordering $<_x$ is the vertex $O_1$ or $O_3$.
Without loss of generality, since $O_1$ and $O_3$ play symmetric roles in $G$, we can assume that $<_x$ is $O_1$.
By Claim~\ref{claim-pairs}, this implies that $T_1$ and $T_2$ form one of these pairs: $\{ABC,ABD\}$, $\{ABD,ACD\}$, $\{ACD,BCD\}$.

We will prove that $\{T_1,T_2\}$ is not equal to $\{ABD,ACD\}$. 
Assume, for a contradiction, that $\{T_1,T_2\}$ is  equal to $\{ABD,ACD\}$.
Let $a,b\in\{y,z\}$ such that $O_2$ is $<_a$ and $O_3$ is $<_b$.
Then, $(ABD,<_x)$ and $(ABD,<_a)$ are edges of $G$. By  Claim~\ref{claim-iff}, $ABD$ is conformal w.r.t. $<_x$ and $<_a$. Since the ordering $<_x$ restricted to $ABD$ is $A <_x B <_x D$, the ordering $<_a$ restricted to $ABD$ is either $A <_a B <_a D$ or $D <_a B <_a A$. Similarly, with the triplet $ACD$, we have either $A <_a C <_a D$ or $D <_a C <_a A$. 
If $A <_a B<_a C <_a D$ or $D <_a C<_a B <_a A$, then the four triplets are conformal w.r.t. $<_x$ and $<_a$, which is a contradiction with the choice of $<_x$ and $<_y$ (they maximize the number of triplets conformal with a pair of orderings), and the assumption of Case 3 (this maximal number equals two).
%If $<_x$ and $<_a$ are equal or reversed on $\{A,B,C,D\}$
%Two triplets of $\E$ are conformal w.r.t. $<_x$ and $<_y$. Hence, $<_x$ and $<_y$ are not equal or reversed on $\{A,B,C,D\}$. 
Therefore, we have either $A <_a C <_a B <_a D$ or $D <_a B <_a C <_a A$. 
On the other hand, the pair $\{T_3,T_4\}$ is equal to $\{ABC,BCD\}$. Since $ABC$ is conformal w.r.t. $<_a$ and $<_b$, the ordering $<_b$ restricted to $ABC$ is either $A <_b C <_b B$ or $B <_b C <_b A$. Similarly, with the triplet $BCD$ we have either $C <_b B <_b D$ or $D <_b B <_b C$. So we have either $A <_b C <_b B <_b D$ or $D <_b B <_b C <_b A$. Hence, the four triplets are conformal w.r.t. $<_a$ and $<_b$,
which is, similarly to above, a contradiction with the maximality property of $<_x$ and $<_y$.
%It is not possible because at most two triplets of $\E$ are conformal w.r.t. $<_y$ and $<_z$. 
So $\{T_1,T_2\}$ is not equal to $\{ABD,ACD\}$.

If $\{T_1,T_2\}$ is equal to $\{ABC,ABD\}$. Figure~\ref{biparti4-ij} illustrates this case. We observe that $BCD$ is conformal w.r.t. $<_y$ and $<_z$, meaning that the configuration induced by $\C$ on $(BCD,\{y,z\})$ is non-fixed. Since we have $A <_x B <_x C <_x D$, $A$ is minimal in $<_x$. Hence, the configuration $\C$ is non-fixed by Lemma~\ref{non-fixe}. If $\{T_1,T_2\}$ is equal to $\{ACD,BCD\}$. We observe that $ABC$ is conformal w.r.t. $<_y$ and $<_z$,
meaning that the configuration induced by $\C$ on $(ABC,\{y,z\})$ is non-fixed. Since we have $A <_x B <_x C <_x D$, $D$ is maximal in $<_x$. Hence, the configuration $\C$ is non-fixed by Lemma~\ref{non-fixe}.

\item
The graph $G$ has at least 2 vertices with degree 3 among $\{O_1,O_2,O_3\}$. See Figure \ref{biparti33}.
\smallskip

By Claim \ref{claim-ij}, $<_x$ and $<_y$ have exactly two common neighbors. We first prove that, without loss of generality, we can assume that $O_1$ is $<_x$, and $O_3$ is $<_y$. First, assume that $O_2$ is $<_x$ or $<_y$ and $O_2$ has degree two.
Then either $O_1$ or $O_3$ has degree four (in order to share two neighbors with $O_2$), implying that $O_1$ and $O_3$ have three common neighbors, which is a contradiction with the maximality property of $<_x$ and $<_y$ and the assumption of Case 3.
So, if $O_2$ has degree two, then we necessarily have $\{<_x,<_y\}=\{O_1,O_3\}$.
Now, if $O_2$ has degree at least three, then we can exchange $O_2$ with $O_1$ or $O_3$ (and some triplets accordingly) in order to have $\{<_x,<_y\}=\{O_1,O_3\}$. Finally, we have $\{<_x,<_y\}=\{O_1,O_3\}$, and  $O_1$ and $O_3$ play symmetric roles in $G$, so we can choose to have $O_1=<_x$ and $O_3=<_y$.

%By Claim \ref{claim-ij}, $O_1$ and $O_3$ have exactly two common neighbors. Hence, without loss of generality, we can assume that $O_1$ is $<_x$, and $O_3$ is $<_y$. Indeed, we necessarily have $\{<_x,<_y\}=\{O_1,O_3\}$ if $O_2$ has degree 2, and we can exchange $O_2$ with $O_1$ or $O_3$ (and some triplets accordingly) if it has degree at least 3.  Then $O_1$ and $O_3$ play symmetric roles in $G$.

%Since the vertex $O_2$ don't share 2 neighbors in common with an other vertex in $\{O_1,O_2,O_3\}$, the vertices $O_1$ and $O_3$ correspond to the orderings $<_x$ and $<_y$. 
%The vertices $O_1$ and $O_3$ are equvalent so we can suppose that $O_1$ correspond to the ordering $<_x$. 
By Claim \ref{claim-pairs}, there are 2 cases: either the vertices $T_2$ and $T_3$ form the pair $\{ACD,BCD\}$, or they form a pair among $\{ABC,ABD\}$ and $\{ABD,ACD\}$.

%Puisque le sommet $O_2$ n'a pas 2 voisins en commun avec un autre sommet correspondant à un ordre, les sommets $O_1$ et $O_3$ correspondent aux ordres $<_x$ et $<_y$. Les sommets $O_1$ et $O_3$ étant équivalents, on peut supposer que $O_1$ correspond à $<_x$. Il y a alors 2 cas à considérer : les sommets $T_2$ et $T_3$ correspondent au couple $\{ACD,BCD\}$ ou ces 2 sommets correspondent à un autre couple ($\{ABC,ABD\}$ ou $\{ABD,ACD\}$).

\begin{figure}[!ht]
\hfil
\begin{minipage}[t]{0.45\linewidth}
\centering \includegraphics[height=4.5cm]{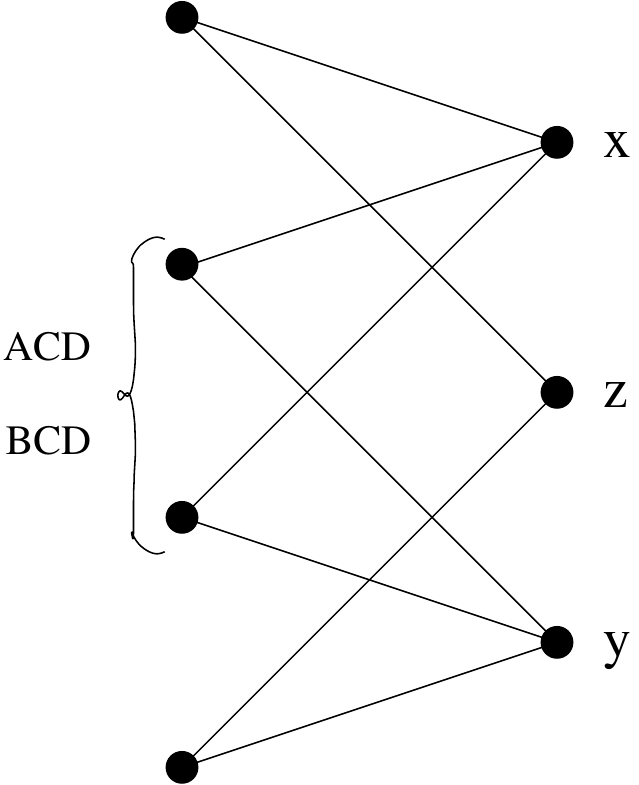}
\caption{$\{T_2,T_3\}=\{ACD,BCD\}$}\label{biparti33-ACD-BCD}
\end{minipage}\hfil
\begin{minipage}[t]{0.45\linewidth}
\centering \includegraphics[height=4.5cm]{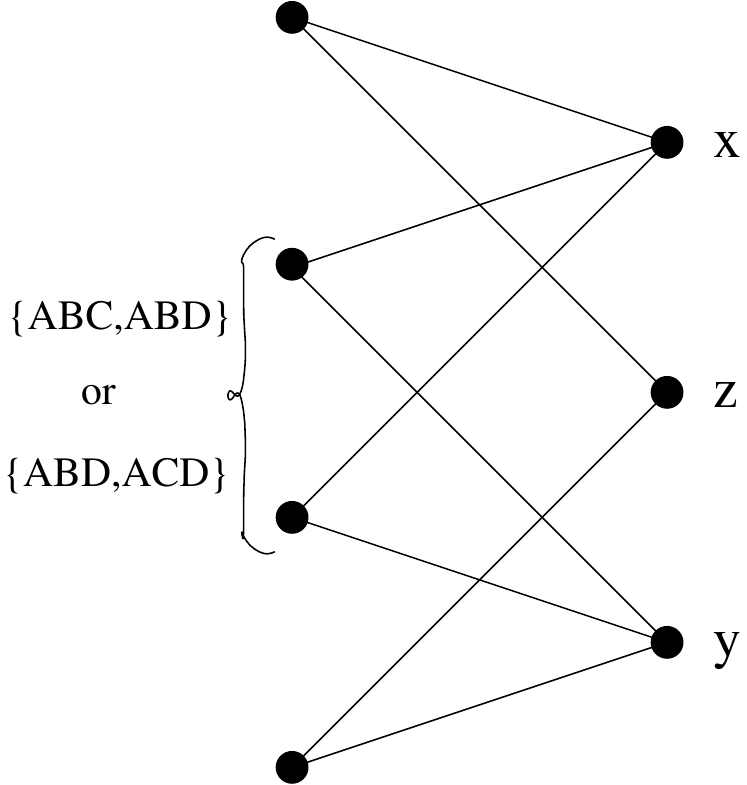}
\caption{$\{T_2,T_3\}\not=\{ACD,BCD\}$}\label{biparti33-autres}
\end{minipage}
\hfil
\end{figure}

%\bigskip

In the first case, we have that $ACD$ and $BCD$ are conformal w.r.t. $<_x$ and $<_y$. Figure~\ref{biparti33-ACD-BCD} shows this case. 
Hence, $ABC$ is conformal w.r.t. $<_z$ and a second ordering $<_a$ for $\{a,b\}=\{x,y\}$. 
%So the projection of $\C$ w.r.t. the third ordering restricted to $ABC$ is non-fixed. 
So the configuration induced by $\C$ on $(ABC,\{z,a\})$ is non-fixed, by Theorem \ref{non-fixe2D}.
Assume $b=x$.
%Assume that this third ordering is $<_x$. 
Since $D$ is extreme in $<_x$, then $\C$ is non-fixed by Lemma~\ref{non-fixe}. 
Now, assume that $b=y$. 
Since $D$ is extremal in $<_x$, and $ACD$ is conformal w.r.t. $<_x$ and $<_y$, then $D$ is also extreme in $<_y$ restricted to $ACD$. 
Similarly, with the triplet BCD, we get that $D$ is extreme in the ordering $<_y$ restricted to $BCD$ and hence extreme in $<_y$.
% We saw that the projection of $\C$ w.r.t. the third ordering restricted to $ABC$ is non-fixed. Since the third ordering is $<_y$ and since $D$ is extremal in the ordering $<_y$, 
So $\C$ is non-fixed by Lemma~\ref{non-fixe}.

%Supposons dans un premier temps que les 2 triplets ordonnés de manière égale ou opposée dans les ordres $<_x$ et $<_y$ soient  $\{ACD,BCD\}$. La Figure~\ref{biparti33-ACD-BCD} représente ce cas. Le triplet $ABC$ est alors ordonnée de manière égale ou opposée dans l'ordre $<_z$ et dans un deuxième ordre (soit $<_x$ soit $<_y$). Ainsi la projection de ${\C}_T$ selon l'ordre correspondant au troisième ordre (soit $<_x$ soit $<_y$) restreinte à $\{A,B,C\}$ est non-fixe. Si ce troisième ordre est $<_x$, $D$ étant extrémal dans $<_x$, le Lemme~\ref{non-fixe} montre que ${\C}_T$ est non-fixe. Si par contre le troisième ordre est $<_y$, comme $ACD$ est ordonné de manière égale ou opposée dans les ordres $<_x$ et $<_y$, le fait que $D$ soit extrémal dans l'ordre $<_x$ nous indique que $D$ est extrémal dans l'ordre $<_y$ restreinte à $ACD$. Comme il en est de même avec le triplet $BCD$ on en déduit que $D$ est extrémal dans l'ordre $<_y$. D'après le Lemme~\ref{non-fixe}, ${\C}_T$ est donc non-fixe.

In the second case,
%Secondly, assume that $(ACD,BCD)$ is not equally ordered or oppositely ordered in the orderings $<_x$ and $<_y$.  
the pair of triplets conformal w.r.t. $<_x$ and $<_y$ is either $\{ABC,ABD\}$ or $\{ABD,ACD\}$. 
 Figure~\ref{biparti33-autres} illustrates this case.
 So $BCD$ is conformal w.r.t. $<_z$ and a second ordering $<_a$, for $\{a,b\}=\{x,y\}$. 
 The configuration induced by $\C$ on $(BCD,\{z,a\})$ is non-fixed, by Theorem \ref{non-fixe2D}.
% The projection of $\C$ w.r.t. the third ordering restricted to $BCD$ is non-fixed. 
Assume $b=x$. 
Since $A$ is extreme in $<_x$, $\C$ is non-fixed by Lemma~\ref{non-fixe}. 
Now, assume $b=y$.
Similarly to our previous demonstration, since $A$ is extreme in $<_x$, we get that $A$ is extreme in $<_y$ when  $\{ABC,ABD\}$ are both conformal w.r.t. $<_x$ and $<_y$, as well as  when $\{ABD,ACD\}$ are both conformal w.r.t. $<_x$ and $<_y$. So $\C$ is non-fixed by Lemma~\ref{non-fixe}.

\end{itemize}
%Supposons maintenant que $(ACD,BCD)$ ne soient pas les 2 triplets qui sont ordonnés de manière égale ou opposée dans les ordres $<_x$ et $<_y$. Alors ces 2 triplets sont soit $(ABC,ABD)$ soit $(ABD,ACD)$. Ainsi $BCD$ est ordonnée de manière égale ou opposée dans l'ordre $<_z$ et dans un deuxième ordre (soit $<_x$ soit $<_y$). La projection de ${\C}_T$ selon l'ordre correspondant au troisième ordre (soit $<_x$ soit $<_y$) restreinte à $\{B,C,D\}$ est alors non-fixe. De même que précédemment, si ce troisième ordre est $<_x$, $A$ étant extrémal dans $<_x$ on peut utiliser le Lemme~\ref{non-fixe} pour montrer que ${\C}_T$ est non-fixe. Si par contre le troisième ordre est $<_y$, on peut montrer comme précédemment que si les 2 triplets qui sont ordonnées de manière égale ou opposée dans les ordres $<_x$ et $<_y$ sont $(ABC,ABD)$ ou $(ABD,ACD)$, alors $A$ étant extrémal dans $<_y$. Le Lemme~\ref{non-fixe} montrer que ${\C}_T$ est non-fixe.

%*** JE PENSE QU'EN UTILISANT LES CLAIMS MIEUX, notammnent le \ref{claim-ij},  ON POURRAIT RACCOURCIR LA PREUVE EN SIMPLIFIANT LES CAS.
\medskip

\noindent{\it Case 4: zero or one triplets of $\E$ are conformal w.r.t. $<_x$ and $<_y$.}
We will prove that this case is not possible.
%\subparagraph*{1 ou 0 triplet}
%We will prove that it is not possible to have less than 2 triplets equally ordered or oppositely ordered in $<_x$ and $<_y$. 
Since there is no triplet of $\E$ such that all the 
configurations induced by $\C$ on this triplet are fixed, 
each triplet is conformal w.r.t.  at least two orderings of $\C$. 
There are three pairs of orderings and four triplets, so there is at least one pair of orderings such that two triplets are conformal w.r.t. this pair of orderings.
\end{proof}

\subsection{Final proofs}

\begin{proof}[Proof of Theorem \ref{caract}]
Recall  that b) implies a) is given by Observation \ref{obs:alg-fixed}.
We have   c) implies b) as noted below Conjecture \ref{conjecture2}: if a configuration induced on a triplet is fixed then it is formally fixed  by Theorem \ref{fixe2D}; as a consequence, formal fixity by expansion for $n=4$ implies formal fixity for $n=4$.
We have d) implies c) by Proposition \ref{prop:proj-fixes}.
%We have c) implies d) by Proposition \ref{tripletfixe} and Proposition \ref{prop:proj-fixes} since formally fixed by projections implies that all projections of $\C$ w.r.t. some triplet are fixed. 
Lastly, to prove that a) implies d), assume d) is false. 
If $\C$ is equivalent to a configuration whose orderings satisfy
\begin{displaymath}
\begin{array}{c c c c c}
B &<_x& C &<_x& A \\
C &<_y& A &<_y& B \\
A &<_z& B &<_z& C
\end{array}
\end{displaymath}
then it is non-fixed by
Propositions~\ref{prop:proj-fixes}. 
If $\C$ is not  equivalent to such a configuration, then, by Proposition \ref{tripletfixe} and a permutation of $\E$, we have that: for every triplet $\E'$ of $\E$ there exists $b\in\B$ such that the configuration induced by $\C$ on $(\E',\B\setminus\{b\})$ is non-fixed.
Then $\C$ is non-fixed by Proposition \ref{proj_non-fixe=>non-fixe}. 
So we find that a) is false, meaning that a) implies d).
\end{proof}

\begin{proof}[Proof of Theorem \ref{caract-non-fixed}]
This proof follows from the proofs of Propositions~\ref{prop:proj-fixes} and \ref{proj_non-fixe=>non-fixe} which enumerate, by means of several cases, all possible non-fixed configurations up to equivalence.
In every case, the fact that a configuration is non-fixed is proved using Lemma \ref{non-fixe}, up to equivalence. 
Hence, every non-fixed configuration is equivalent to a configuration satisfying the hypothesis of this lemma.
%Since, obviously, permutations of $\B$ or $\E$ and ordering reversions do not change this property,
Since permutations of $\B$ or $\E$ and ordering reversions obviously do not change this property, we get that every non-fixed configuration satisfies the hypothesis of this lemma.
\end{proof}

%%%%%%%%%%%%%%%%%%%%%%%%%%%%%%%%%%%%%%%%%%%%%%%%%%%%%%%%%%
%%%%%%%%%%%%%%%%%%%%%%%%%%%%%%%%%%%%%%%%%%%%%%%%%%%%%%%%%%

%\bibliographystyle{abbrv}
\bibliographystyle{plain}

\begin{thebibliography}{99}

\bibitem{OM99}
A. Bj\"orner, M. Las Vergnas, B. Sturmfels, N. White, G. Ziegler,
\newblock {\it {Oriented matroids} }
\newblock 2nd ed., Encyclopedia of Mathematics and its Applications 46, 
Cambridge University Press, Cambridge, UK 1999.

\bibitem{Crane}
J. Braga, J. Treil.
{\it Estimation of pediatric skeletal age using geometric
morphometrics and three-dimensional cranial size changes.}
%Int. J. Legal. Med. (2007) 121:439--443.
Int. J. Legal. Med. (2007) 121:439--443.

%V23 ajout de ref ci-dessous
\bibitem{SNS}
R. Brualdi, B. Shader.
{\it Matrices of Sign-Solvable Linear Systems.} 
Cambridge University Press, Cambridge, UK 1995. 

%V23 ajout de ref ci-dessous
\bibitem{these}
K. Sol. 
\newblock
{\it Une approche combinatoire novatrice fondée sur les matroïdes orientés pour la caractérisation de la morphologie 3D des structures anatomiques.}
\newblock
Ph. D. Dissertation, Universit\'e Montpellier 2, France 2013.

\bibitem{cccg}
E. Gioan, K. Sol, G. Subsol.
Orientations of Simplices Determined by Orderings 
on the Coordinates of their Vertices.
Proceedings CCCG'2011 (Canadian Conference on Computational Geometry), 6p. (2011).
Short preliminary conference version of the present paper.

\bibitem{oeisf}
E. Gioan, K. Sol, G. Subsol.
Sequence A201973 in The On-Line Encyclopedia of Integer Sequences (2011). Published electronically at 
%http://oeis.org
http://oeis.org/A201973

%V23 ajout de ref ci-dessous
\bibitem{miccai}
E. Gioan, K. Sol, G. Subsol.
{\it A Combinatorial Method for 3D Landmark-based Morphometry: Application to the Study of Coronal Craniosynostosis.}
Proceedings MICCAI 2012 (Medical Image Computing and Computer-Assisted Intervention), LNCS, 2012, Volume 7512/2012, 533-541. 
%COURT
%\bibitem{poster}
%E. Gioan, K. Sol, G. Subsol, Y. Heuz\'e, J. Richtsmeier, J. Braga, F. Thackeray.
%\newblock
%A new 3D morphometric method based on a combinatorial encoding of 3D point configurations:
%application to skull anatomy for clinical research and physical anthropology.
%\newblock
%Poster presented at the 80th Annual Meeting of the American Association of Physical Anthropologists (2011).
%%(Minneapolis,  April 12-16, 2011).

%COMPLET
\bibitem{poster}
E. Gioan, K. Sol, G. Subsol, Y. Heuz\'e, J. Richstmeier, J. Braga, F.
Thackeray. 
\newblock
{\it A new 3D morphometric method based on a combinatorial encoding
of 3D point configurations: application to skull anatomy for clinical
research and physical anthropology.}
\newblock
Poster: 80th Annual Meeting of the American
Association of Physical Anthropologists, Minneapolis (U.S.A.), April 2011.
\newblock
Abstract: American Journal of Physical Anthropology, p. 280,
Vol. 144 Issue S52, 2011.

\end{thebibliography}

%No reference is really needed in this paper which relies on classical mathematics and which addresses a problem that the authors did not find in the literature.

\end{document}